%% file: main.tex
\documentclass[11pt]{article}
\usepackage{style}
\usepackage{shortcuts}
\pdfoutput=1

\author{Amir Abboud\thanks{Weizmann Institute of Science. This work is part of the project CONJEXITY that has received funding from the European Research Council (ERC) under the European Union's Horizon Europe research and innovation programme (grant agreement No.~101078482). Supported by a research grant from the Center for New Scientists at the Weizmann Institute of Science. Parts of this work were done while visiting INSAIT, Sofia University ``St. Kliment Ohridski'', Rutgers University, and the Courant Institute of Mathematical Sciences at New York University.} \and Ron Safier\thanks{Weizmann Institute of Science.} \and Nathan Wallheimer\thanks{Weizmann Institute of Science.}}

\date{}

\title{Triangle Detection in $H$-Free Graphs}

\begin{document}

    \maketitle 

    \begin{abstract}
        We initiate the study of combinatorial algorithms for Triangle Detection in $H$-free graphs. The goal is to decide if a graph that forbids a fixed pattern $H$ as a subgraph contains a triangle, using only ``combinatorial'' methods that notably exclude fast matrix multiplication. Our work aims to classify which patterns admit a subcubic speedup, working towards a dichotomy theorem. 

        On the lower bound side, we show that if $H$ is not $3$-colorable or contains more than one triangle, the complexity of the problem remains unchanged, and no combinatorial speedup is likely possible.

        On the upper bound side, we develop an embedding approach that results in a strongly subcubic, combinatorial algorithm for a rich class of ``embeddable'' patterns. 
        Specifically, for an embeddable pattern of size $k$, our algorithm runs in $\tOrder(n^{3-\frac{1}{2^{k-3}}})$ time, where $\tOrder(\cdot)$ hides poly-logarithmic factors. 
        This algorithm also extends to listing all the triangles within the same time bound. We supplement this main result with two generalizations:
        \begin{itemize}
            \item A generalization to patterns that are embeddable up to a single obstacle that arises from a triangle in the pattern. This completes our classification for small patterns, yielding a dichotomy theorem for all patterns of size up to eight. 
            \item An $H$-sensitive algorithm for embeddable patterns, which runs faster when the number of copies of $H$ is significantly smaller than the maximum possible $\Omega(n^{k})$. 
        \end{itemize}
        Finally, we focus on the special case of odd cycles. We present  specialized Triangle Detection algorithms that are very efficient: 
        \begin{itemize}
            \item A combinatorial algorithm for $C_{2k+1}$-free graphs that runs in $\tOrder(m+n^{1+2/k})$ time for every $k \geq 2$, where $m$ is the number of edges in the graph. 
            \item A combinatorial $C_5$-sensitive algorithm that runs in $\tOrder(n^2 + n^{4/3} t^{1/3})$ time, where $t$ is the number of $5$-cycles in the graph. 
        \end{itemize}
    \end{abstract} 

    \thispagestyle{empty} 
    \newpage 
    \thispagestyle{empty} 
    \tableofcontents
    \newpage
    \setcounter{page}{1} 

    \section{Introduction}
        \label{sec:introduction}
        \input{sections/introduction.tex}

    \section{Technical Overview}
        \label{sec:overview}
\input{sections/overview.tex}
    \section{Related Works}
        \label{sec:related}
        \input{sections/related.tex}

    \section{Preliminaries}
        \label{sec:preliminaries}
        \input{sections/preliminaries.tex}

    \section{Lower Bounds}
        \label{sec:lowerbounds}
        \input{sections/lowerbounds.tex}

    \section{An Algorithm for Nicely Colored Patterns}
        \label{sec:H-nice}
        \input{sections/H-nice.tex}

    \section{$H$-Sensitive Algorithms}
        \label{sec:sensitive}
        \input{sections/sensitive.tex}
    \section{An Algorithm for Patterns with an Attached Triangle}
        \label{sec:attached}
        \input{sections/attached.tex}
    \section{An Algorithm for $C_{2k+1}$-Free Graphs}
        \label{sec:odd}
        \input{sections/odd.tex}

    \section{Proofs of Propositions}
        \label{sec:propositions}
        \input{sections/propositions.tex}
    \section{Conclusions}
        \label{sec:conclusions}
        \input{sections/conclusions.tex}
    \section*{Acknowledgments}
        We would like to thank Elad Tzalik for useful discussions.

    \newpage

    \bibliographystyle{plainurl}
    \bibliography{references}
\end{document}

%% file: sections/introduction.tex
In the Triangle Detection problem, a graph $G$ is given as input, and the goal is to decide whether $G$ contains a triangle, i.e., three vertices such that an edge connects every pair of them. 
The trivial algorithm that checks every triplet of vertices runs in $\Order(n^3)$ time, and the only other known algorithm for the problem achieving strongly subcubic, i.e., $\Order(n^{3-\eps})$ time for some $\eps>0$, is the  $\Order(n^\omega)$ time algorithm from~\cite{itai1977finding} via fast matrix multiplication, where $\omega < 2.371339$~\cite{alman2025more}. 
So-called \emph{combinatorial} algorithms, which avoid fast matrix multiplication, have achieved only sub-polynomial improvements so far despite decades-long efforts~\cite{arlazarov1970economical,BW12,chan2014speeding,yu2018improved,abboud2024new}, leading to the conjecture that $n^{3-o(1)}$ time is unavoidable for combinatorial algorithms.
A remarkable reduction by Vassilevska and Williams~\cite{williams2010subcubic} shows that this conjecture is equivalent to the assumption that no strongly subcubic combinatorial algorithm exists for Boolean Matrix Multiplication (BMM). 

\begin{conjecture}[Combinatorial Triangle Detection (equivalently, the BMM Conjecture)]
\label{conj:bmm}
There is no strongly subcubic combinatorial Triangle Detection (equivalently, BMM) algorithm.
\end{conjecture}

Triangle Detection is one of the most fundamental problems in algorithmic graph theory and fine-grained complexity; not only is it the smallest non-trivial case of $k$-Clique Detection and $k$-Cycle Detection, making it arguably the simplest graph problem we cannot solve in linear time, but its hardness also explains the complexity of a vast number of other problems. Indeed, many of the central conjectures in fine-grained complexity, which form the basis for dozens of conditional lower bounds, can be viewed as variants, strengthenings, or generalizations of the assumption that Triangle Detection is a hard problem. 
In particular, as discussed by Abboud et al.~\cite[Section 1.1]{abboud2024new}, a refutation of the Combinatorial Triangle Detection conjecture would be a major step towards refuting central conjectures such as All-Pairs Shortest-Paths (equivalent to the weighted variant), 3SUM (related to the reporting variant), and Online Matrix Vector Multiplication (related to the online variant). 
This motivates the study of this particular problem for two main reasons. First, for the purpose of gaining a better understanding of what makes problems hard, e.g., informing us about the hard instances of all these problems. Second, for the purpose of attacking the conjectures of fine-grained complexity, which should be viewed as crisply articulating the edges of our algorithmic capabilities, and hence being the most pressing for algorithmic research. 

An inspiring approach towards attacking Conjecture~\ref{conj:bmm} is via the regularity-based techniques in the $n^{3}/\Omega(\log^{2.25} n)$ time algorithm of Bansal and Williams~\cite{BW12} and in the recent $n^3/2^{\Omega(\log^{1/7}n)}$-time algorithm of Abboud et al.~\cite{abboud2024new}. 
Both works employ the ``structure vs. randomness'' methodology~\cite{tao2007structure}, which exploits a dichotomy between structure and pseudorandomness by decomposing the input into a structured part and a pseudorandom part. The precise definitions of structure and pseudorandomness can vary, as worst-case instances of Triangle Detection can be highly structured under one definition while being highly pseudorandom under another. For example, the neighborhood of every vertex in a triangle-free graph is an independent set (a highly structured property). However, these graphs can also be good expanders (a pseudorandom property)~\cite{abboud2023worst,abboud2024worst}. The currently known structural properties of worst-case instances (e.g., large independent sets) are not known to be algorithmically useful for recognizing triangle-free graphs efficiently. A challenge towards algorithmic progress is in finding the right notion of pseudorandomness for this context.

In this work, we aim to improve the understanding of the structure of worst-case Triangle Detection instances. A better understanding of this question is crucial for such a fundamental problem and may direct us toward resolving \cref{conj:bmm}. We initiate the methodological study of combinatorial Triangle Detection algorithms for structured inputs. In the language of algorithmic graph theory, we consider the following question. 

\begin{center}
\emph{
Which families of graphs admit a strongly subcubic combinatorial Triangle Detection algorithm?}
\end{center}

For graph families in which every graph is sparse, namely, has $m = \Order(n^{2-\eps})$ edges for some $\eps > 0$, the answer is trivial; a straightforward $\Order(mn)$ algorithm is already strongly subcubic. Such graph families include bounded-arboricity graphs, minor-free graphs, bounded treewidth graphs, and more. In this paper, we focus on the family of $H$-free graphs for a fixed pattern $H$: graphs that do not contain the pattern $H$ as a subgraph. This family of graphs is an excellent candidate for studying the properties of structured inputs. For example, $H$-free graphs can be dense, which implies that the following question cannot be answered by only exploiting sparsity. 

\begin{openquestion}
\label{q:free}
Which patterns $H$ admit a strongly subcubic combinatorial Triangle Detection algorithm for $H$-free graphs? 
\end{openquestion}

$H$-free graphs are also central objects in extremal graph theory, and the number of copies of a pattern $H$ is often used as a measure of pseudorandomness. In particular, the Chung-Graham-Wilson theorem~\cite{chung1989quasi} shows that a graph in which the count of every pattern $H$ is as expected in a random graph with the same density (up to an additive error) satisfies many equivalent notions of pseudorandomness, including strong spectral expansion. It is therefore natural to ask whether a graph containing few copies of $H$ has enough structure to be exploited for recognizing triangle-free graphs efficiently. Throughout the paper, we use $k$ to denote the number of vertices in $H$. 

\begin{openquestion}
\label{q:sensitive}
For which patterns $H$ is there a combinatorial Triangle Detection algorithm such that if a graph $G$ contains at most $n^{k-\eps}$ copies of $H$ for some $\eps > 0$, the algorithm runs in at most $n^{3-\delta}$ time for some $\delta := \delta(\eps) > 0$?
\end{openquestion}

In this paper, we answer both of these questions for large and complex families of patterns. 
We begin with the first question about $H$-free graphs, which are the main focus of the paper. 

\subsection{Our results for Triangle Detection in $H$-free graphs}
On the lower bound side, we identify two classes of patterns for which, under \cref{conj:bmm}, there is no strongly subcubic Triangle Detection algorithm for $H$-free graphs. Formally, we prove the following theorem using standard self-reductions in this branch of research.

\begin{restatable}[Lower bounds classes]{theorem}{lowerbounds}
\label{thm:lowerbounds}
If $H$ is not $3$-colorable, or if it contains more than one triangle, then combinatorial Triangle Detection algorithms for $H$-free graphs require $n^{3-o(1)}$ time under \cref{conj:bmm}.
\end{restatable}

On the upper bound side, we employ three different techniques to solve the problem for different classes of patterns, putting us close to proving a \emph{dichotomy theorem}: a complete classification of patterns into those for which Triangle Detection in $H$-free graphs is strongly subcubic, and those for which a conditional lower bound of $n^{3-o(1)}$ exists.

Before we present the results, let us build some intuition. Our starting point is bipartite patterns. Suppose $H$ is a bipartite graph with parts of sizes $s \leq t$. By the Kov{\'a}ri-S{\'o}s-Tur{\'a}n theorem~\cite{kovari1954problem} (abbreviated as KST), an $H$-free graph $G$ contains at most $\Order(n^{2-1/s})$ edges, so the problem can be solved trivially in $\Order(n^{3-1/s})$ time. 

For non-bipartite graphs, the smallest case is $H=K_3$, which is trivial: If $G$ is triangle-free, an algorithm can answer NO in constant time. Another simple case is when $H$ is a triangle with an attached edge. In an $H$-free graph, any vertex of a triangle cannot have a neighbor outside that triangle. This implies that the vertices of a triangle in $G$ have degree $2$. Therefore, one can find a triangle in $\Order(m+n)$ time by processing only the vertices of degree $2$. The smallest non-trivial case is when $H$ is a cycle of length $5$, denoted by $C_5$. Our starting point is a combinatorial and strongly subcubic Triangle Detection algorithm for $C_5$-free graphs.
\begin{proposition}
    \label{prop:C5-warmup}
    There is a combinatorial Triangle Detection algorithm for $C_5$-free graphs, running in time $\tOrder(n^{7/3})$ and succeeding with probability $1-\Order(1/n)$. 
\end{proposition}
In later sections, we will see that this algorithm is not optimal; however, it serves as an important warm-up because it utilizes many of the core ideas in our most general results. Specifically, the approach is to try to embed the forbidden pattern randomly in the graph $G$. The randomness ensures that a triangle that is incident to a high-degree vertex is contained in the area surrounding the embedded pattern. Eventually, we gain a speed-up by restricting the search to an area that, on the one hand, contains a triangle edge. On the other hand, every non-triangle edge in it completes the embedding of the forbidden pattern. Essentially, the problem is reduced to deciding if there is an edge in a subgraph of $G$. 

By generalizing this embedding approach, we provide an algorithm that applies to more patterns. It turns out that there is a large and rich family of patterns that can be embedded in worst-case Triangle Detection graphs, resulting in a polynomial speed-up. 
These \textit{embeddable} patterns are characterized by admitting a special type of $3$-coloring, defined as follows. 
\begin{restatable}[Nice colorings]{definition}{nice}
\label{def:nice}
    Let $C:V(H) \mapsto \set{\texttt{RED},\texttt{BLUE},\texttt{GREEN}}$ be a proper $3$-coloring of $H$. We say that $C$ is \emph{nice} if, by repeatedly deleting vertices with monochromatic neighborhoods until no such vertices remain, the leftover graph is either an empty graph or a triangle. We may also say that $H$ is nicely colored by $C$. 
\end{restatable}
\noindent
For patterns that admit a nice coloring, we provide the following result.
\begin{restatable}{theorem}{algnice}
    \label{thm:H-nice}
    For every pattern $H$ that admits a nice coloring, there is a combinatorial Triangle Detection algorithm for $H$-free graphs, running in time:
    \begin{equation*}
        \begin{cases}
        \tOrder(n^{3 - \frac{1}{2^{k-3}}}) & \text{ if $H$ is triangle-free,} \\
        \tOrder(n^{3 - \frac{1}{2^{k-4}}}) & \text{ if $H$ contains a triangle.} 
    \end{cases}
    \end{equation*}
    The algorithm succeeds with probability $1-\Order(1/n)$.
\end{restatable}
An immediate question that arises is about the generality of this result, i.e., which patterns admit a nice coloring. Let us begin by showing that it captures the most basic cases already discussed.
First, observe that all bipartite graphs admit a nice coloring, because every neighborhood in a proper $2$-coloring must be monochromatic. A triangle clearly admits a nice coloring. Next, observe that every cycle of odd length greater than $3$ admits a nice coloring, e.g., by coloring it greedily with two colors until the process is stuck at the last vertex, which gets a third color. The reason this coloring is nice is that the second vertex to receive a color will have a monochromatic neighborhood. After removing it, the remaining pattern is a path, and its endpoints must also have a monochromatic neighborhood. As an immediate consequence, we also get that all patterns which are $C_{k}$-colorable\footnote{A graph $G$ is $H$-colorable if there is a homomorphism from $G$ to $H$. Equivalently, it means that $G$ is a subgraph of a blowup of $H$, i.e., the graph obtained by replacing every edge of $H$ with a biclique.} for some $k > 3$ admit a nice coloring, as the greedy coloring can also nicely color the color classes in a $C_{k}$-colored pattern. However, the result is even more general than $C_k$-colored patterns: For example, the smallest triangle-free patterns that are not $C_5$-colorable have $8$ vertices~\cite[Theorem 1.1, Figure 3]{goedgebeur2024minimal}. The graph in \cref{fig:obfuscation} is such a pattern, yet it also admits a nice coloring that we illustrate in the figure.
 \begin{figure}[h]
    \centering
    \includegraphics[scale=0.8]{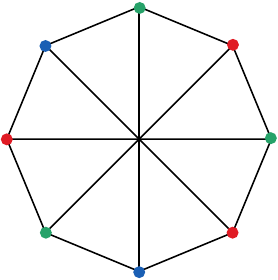}
    \caption{An $8$-vertex triangle-free graph that is not $C_5$-colorable, but admits a nice coloring (illustrated).}
    \label{fig:obfuscation}
\end{figure}

Perhaps every pattern $H$ which does not fall in the lower bound classes of \cref{thm:lowerbounds} admits a nice coloring? From a local view, it seems plausible: If $H$ is $3$-colorable and contains at most one triangle, then the neighborhoods of non-triangle vertices are independent sets. Thus, these neighborhoods can be properly colored with only one color. A similar notion of colorings, called adynamic colorings, was studied by Šurimová et al.~\cite{vsurimova2020adynamic}. This paper refutes that idea by giving an example of a triangle-free, $3$-colorable graph $H$, such that every proper $3$-coloring of $H$ does not admit a vertex with a monochromatic neighborhood. Specifically, $H$ consists of two copies of a $6$-cycle such that each vertex in the first copy is connected to two vertices in the second copy: to its copy and to the antipodal vertex of its copy. See \cref{fig:12vertices}.
    
 \begin{figure}[h]
    \centering
    \includegraphics[scale=0.8]{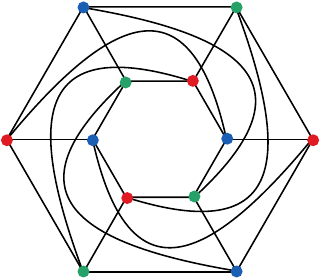}
    \caption{The graph from~\cite[Figure 6]{vsurimova2020adynamic}, together with a proper $3$-coloring (illustrated). This graph does not admit a nice coloring, and in every proper $3$-coloring of it, the neighborhood of every vertex is two-colored.}
    \label{fig:12vertices}
\end{figure}

A curious observation that we make is that cycles of length divisible by $3$ are inherent in such examples. In particular, if a $3$-colorable graph is free from such cycles, then not only does it admit a nice coloring, but \emph{every} proper $3$-coloring is nice.
\begin{restatable}{proposition}{divisible}
    \label{prop:3k}
    Let $H$ be a $3$-colorable, $3k$-cycle free graph for every $k \in \Nat$. Then in every proper $3$-coloring of $H$ there is a vertex with a monochromatic neighborhood.
\end{restatable}

Understanding these examples is crucial to our goal of proving a dichotomy theorem. It could be that the missing component for such a theorem is a lower bound for patterns that do not admit a nice coloring. 
It turns out that the story does not end here, as there are patterns $H$ which do not admit a nice coloring, yet $H$-free graphs admit a strongly subcubic combinatorial Triangle Detection algorithm.

Our next family of patterns for which we show a strongly subcubic algorithm is a generalization to patterns that are essentially characterized by a single triangle being the only obstacle to a nice coloring. 
\begin{restatable}{theorem}{attached}
\label{thm:attached}
    Suppose that $H$ satisfies the following property. There is a $3$-coloring of $H$, such that after removing $s$ vertices with monochromatic neighborhood until no such vertices remain, the leftover pattern consists of a triangle $xyz$ and a subpattern $H'$, with $xy \in E(H')$ and $z \notin V(H')$, and $H'$ is nicely colored. 
    Then there is a combinatorial Triangle Detection algorithm for $H$-free graphs that runs in time $\tOrder\left(n^{3-\frac{1}{2^{k + s - 1}}}\right)$. The algorithm succeeds with probability $1-\Order(1/n)$. 
\end{restatable}

\noindent
Together, our results provide a complete picture for small patterns.
\begin{restatable}{proposition}{smallpatterns}
    \label{prop:smallpatterns}
    Every pattern $H$ with at most $8$ vertices falls into one of the following four categories:
    \begin{itemize}
        \item Trivially Solvable: The graph is a subpattern of a triangle, and Triangle Detection can be solved in constant time.
        \item Conditional Lower Bound: The pattern is not $3$-colorable or contains more than one triangle, implying no strongly subcubic combinatorial algorithm exists under \cref{conj:bmm}.
        \item Nice Coloring: The pattern admits a nice coloring, and a strongly subcubic combinatorial algorithm is provided by \cref{thm:H-nice}.
        \item Attached Triangle: The pattern falls under the classification in \cref{thm:attached}, and a strongly subcubic combinatorial algorithm is provided by the theorem.
        \end{itemize}
\end{restatable}
\noindent
We also suggest a minimal example of a pattern $H$ that does not fall into any of the four categories. The pattern is constructed similarly to the one in \cref{fig:12vertices}, but we replace one of the $6$-cycles with a triangle. See \cref{fig:9vertices}.
     \begin{figure}[h]
        \centering
        \includegraphics[scale=0.8]{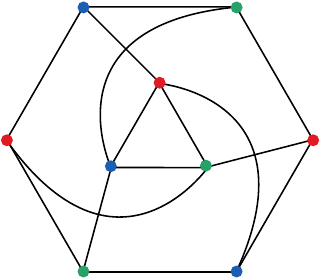}
        \caption{A minimal example of a pattern which does not fall into any of the categories of \cref{prop:smallpatterns}. Namely, the pattern is $3$-colorable (illustrated) and contains exactly one triangle. In every proper $3$-coloring of it, there are no vertices with a monochromatic neighborhood. Moreover, every vertex of the triangle has degree greater than $2$ and therefore it does not fall into the classification of \cref{thm:attached}.}
        \label{fig:9vertices}
    \end{figure}

\subsection{$H$-sensitive algorithms}
Another set of results we provide addresses \cref{q:sensitive} about $H$-sensitive algorithms. These are algorithms that run faster when the number of copies of a pattern $H$ is polynomially smaller than the maximum possible, which can be as high as $\Omega(n^k)$. $H$-sensitive algorithms are desirable because their complexity improves upon the worst-case complexity of the problem (i.e., they are never slower). In particular, an $H$-sensitive algorithm demonstrates that the hardest instances of Triangle Detection contain a maximum number of copies of $H$, aligning with the idea that pseudorandom graphs are the most challenging. 
We provide a generalization of the algorithm from \cref{thm:H-nice} to an $H$-sensitive algorithm for every $H$ that admits a nice coloring.

\begin{restatable}{theorem}{Hsensitive}
\label{thm:H-sensitive}
For every pattern $H$ that admits a nice coloring, there is a combinatorial algorithm that, given a graph $G$ and a value $t$ which upper bounds the number of copies of $H$ in $G$, decides whether $G$ contains a triangle. The running time is given by:
\[
    \begin{cases}
        \tOrder\left( n^{3-\frac{1}{2^{k-1}}} + n^{3 - \frac{k}{2^k - 1}} \cdot t^{\frac{1}{2^k - 1}}\right)  & \text{if } H \text{ is triangle-free}, \\
        \tOrder\left(n^{3-\frac{1}{2^{k-4}}} + n^{3-\frac{k-3}{2^{k-3}-1}} \cdot t^{\frac{1}{2^{k-3}-1}}\right) & \text{if } H \text{ contains a triangle}.
    \end{cases}
\]
The algorithm succeeds with probability $1-\Order(1/n)$. 
\end{restatable}
Note that the distinction between triangle-free patterns and those that contain a triangle is necessary: By \cref{thm:lowerbounds}, there exist worst-case Triangle Detection instances with only one triangle, and consequently, only $\Order(n^{k-3})$ copies of $H$ for every pattern $H$ that contains a triangle, as opposed to potentially $\Omega(n^k)$ copies in the triangle-free case. Therefore, an $H$-sensitive algorithm for a pattern $H$ that contains a triangle should run in cubic time when $t = \Omega(n^{k-3})$, whereas if $H$ is triangle-free, then the running time for the same value of $t$ should be strongly subcubic. 

We also remark that the algorithm requires an upper bound $t$ on the number of copies of $H$ as input. This assumption slightly weakens the result, but it can be removed easily in the case where $H$ is triangle-free by computing an (not necessarily tight) upper bound on the number of copies of $H$ during a preprocessing step, which still yields a strongly subcubic algorithm if $G$ has polynomially fewer than $n^k$ copies of $H$.\footnote{For example, consider a triangle-free pattern $H$. Suppose that we take $\Theta(n^2 \log n)$ samples of $k$-tuples of vertices, and for each tuple we check if it induces a copy of $H$. If the first copy of $H$ was found after $ \Omega(n^{\alpha} \log n )$ samples, for any $0 < \alpha \leq 2$, we conclude that $t \leq n^{k-\alpha}$ with high probability. If no copy appears in $\Theta(n^2 \log n)$ samples, we conclude that $t \leq n^{k-2}$.} We note that if $H$ contains a triangle, then it is necessary for an $H$-sensitive algorithm to receive $t$ as an input, under \cref{conj:bmm}: If there was an $H$-sensitive Triangle Detection algorithm that is oblivious to the upper bound $t$, then it would need to run in subcubic time on every triangle-free graph $G$ (as then $t = 0$). Using such an algorithm, we can distinguish triangle-free graphs from graphs that contain a triangle by running the supposed algorithm and answering YES if its running time exceeds the stated bound when $t = 0$. 

Additionally, we provide a more efficient algorithm for the special case of $H = C_5$.
\begin{restatable}{theorem}{Cfivesensitive}
    \label{thm:C5-sensitive}
    There is a combinatorial Triangle Detection algorithm that runs in $\tilde{\Order}(n^{2} + n^{4/3} t^{1/3})$ time, where $t$ is the number of $5$-cycles in the input graph. The algorithm succeeds with probability at least $1-\Order(1/n)$.
\end{restatable}
\noindent 
Note that, unlike the previous algorithm, this one is oblivious to $t$.

\subsection{A special focus on $C_{2k+1}$-free graphs}
    Cycles (using $C_k$ to denote the cycle of length $k$) are one of the most natural classes of graphs, so a natural question is about the running time of Triangle Detection in $C_{k}$-free graphs, for each $k$. 
    For the case of even cycles, $C_{2k}$ is bipartite, so the KST theorem can be employed to get an $\Order(n^{3-1/k})$ time algorithm. In fact, it is possible to do much better than that, because not only is the graph sparse, it follows from the Bondy-Simonovits theorem that a $C_{2k}$-free graph is $d$-degenerate for $d = \Order(n^{1/k})$~\cite{bondy1974cycles}. A graph is $d$-degenerate if every induced subgraph contains a vertex of degree less than $d$. The classic Chiba-Nishizeki algorithm solves Triangle Detection in $d$-degenerate graphs in $\Order(md) = \Order(n^{1+2/k})$ time~\cite{chiba1985arboricity}.\footnote{The original paper uses arboricity rather than degeneracy, but the two parameters are known to differ by at most a factor of two.} For odd cycles, the problem is more complex, as illustrated by complete bipartite graphs, which are $C_{2k+1}$-free but also very dense. We overcome this challenge by essentially restricting the search to graphs that have the form of a dense bipartite graph with a bounded degeneracy graph added to one of its parts. This requires a generalization of the Chiba-Nishizeki algorithm to graphs that are only partially degenerate. In addition, we employ a novel idea, an algorithmic adaptation of a supersaturation result for even cycles by Jiang and Yepremyan~\cite{jiang2020supersaturation}. Together, we manage to prove the following theorem. 

    \begin{restatable}{theorem}{odd}
        \label{thm:odd}
        For every $k \geq 2$, there is a combinatorial Triangle Detection algorithm for $C_{2k+1}$-free graphs that runs in $\tOrder( m + n^{1+2/k})$ time. 
    \end{restatable}
    \noindent
    Note that this running time matches the bound for $C_{2k}$-free graphs (using Chiba-Nishizeki), up to an additive $m$ term.

    A summary of the results so far is presented in \cref{table:results}.

\begin{center}
\begin{table}[h!]
\caption{Our results for Triangle Detection in $H$-free graphs.}
\label{table:results}
\footnotesize
\begin{tabular}{ |c|c|c|c| }
\hline
\textbf{Classification} & \textbf{Condition on $H$} & \textbf{Bound} & \textbf{Theorem} \\
\hline
\multirow{2}{12em}{Lower bound under the BMM Conjecture} 
& Not $3$-colorable & $n^{3-o(1)}$ & \labelcref{thm:lowerbounds} \\ 
\cline{2-4}
& Contains more than one triangle & $n^{3-o(1)}$ & \labelcref{thm:lowerbounds} \\ 
\hline
\multirow{4}{12em}{\vfill Subcubic upper bound} 
& $H = C_{2k+1}$ ($k \ge 2$) & $\tilde{O}(m + n^{1+2/k})$ & \labelcref{thm:odd}\\ 
\cline{2-4}
& Nicely colored and triangle-free & $\tilde{O}(n^{3-\frac{1}{2^{k-3}}})$ & \labelcref{thm:H-nice} \\ 
\cline{2-4}
& Nicely colored and contains a triangle  & $\tilde{O}(n^{3-\frac{1}{2^{k-4}}})$ & \labelcref{thm:H-nice} \\ 
\cline{2-4}
& \makecell{``Attached Triangle'' \\($s = \#$ removable vertices)} & $\tilde{O}(n^{3-\frac{1}{2^{k+s-1}}})$ &  \labelcref{thm:attached} \\
\hline
\multirow{2}{12em}{\vfill Subcubic upper bound for any $t < n^{k-\eps}$ ($t = \# H$)} 
& $H = C_5$ & $\tilde{O}(n^{2} + n^{4/3}t^{1/3})$ & \labelcref{thm:C5-sensitive}\\ 
\cline{2-4}
& Nicely colored and triangle-free & \makecell[l]{$\tilde{O}(n^{3-\frac{1}{2^{k-1}}} $\\$\quad +n^{3-\frac{k}{2^{k}-1}} \cdot t^{\frac{1}{2^{k}-1}})$} & \labelcref{thm:H-sensitive}\\

\cline{2-4}
& Nicely colored and contains a triangle & \makecell[l]{$\tilde{O}(n^{3-\frac{1}{2^{k-4}}} $\\$\quad +n^{3-\frac{k-3}{2^{k-3}-1}} \cdot t^{\frac{1}{2^{k-3}-1}})$} & \labelcref{thm:H-sensitive}\\
\hline
\end{tabular}
\end{table}
\end{center}

\subsection{Generalized Tur\'{a}n problem and Triangle Listing}

A related line of work has considered the purely combinatorial question of bounding the number of triangles in $H$-free graphs. This Tur\'{a}n-type question was initiated in a paper by Alon and Shikhelman~\cite{alon2016many}, and since its introduction, many papers on the subject have appeared (see~\cite{gerbner2025survey} for a recent survey). For every pattern $H$, let $ex(K_3,H)$ be the maximum number of triangles in an $n$-vertex $H$-free graph. An analogous result to our first lower bound class in \cref{thm:lowerbounds} follows from~\cite[Proposition 2.1]{alon2016many} and shows that $ex(K_3,H)$ is strongly subcubic if and only if $H$ is $3$-colorable. 
One might wonder if the proof is algorithmic, in the sense that it also provides an efficient way to find the triangles if $H$ is $3$-colorable. However, an algorithmic proof of this result cannot hold generally for all $3$-colorable patterns. That is because there are patterns that are $3$-colorable but do not admit a strongly subcubic combinatorial Triangle Detection algorithm under \cref{conj:bmm}. For example, if $H$ is the diamond graph (two triangles that share an edge), then clearly $ex(K_3,H) = \Order(n^2)$, whereas \cref{thm:lowerbounds} shows that Triangle Detection in diamond-free graphs requires $n^{3-o(1)}$ time under \cref{conj:bmm}. In the opposite direction, the algorithmic question may contribute to the mathematical one, as a combinatorial Triangle Detection algorithm for $H$-free graphs offers insight into the structure of triangles in $H$-free graphs, which may be used to prove upper bounds on $ex(K_3,H)$. 

A more relatable variant of Triangle Detection is the Triangle Listing problem, in which the goal is to produce a list of all the triangles in a graph. This problem generalizes not only Triangle Detection but also its weighted variants, such as Negative Triangle (an APSP-hard problem~\cite{williams2010subcubic}) and $0$-Triangle (a 3SUM-hard problem~\cite{vassilevska2009finding}). A primary objective in designing a listing algorithm is to ensure it is output-sensitive, meaning the running time is proportional to the number of triangles in the graph. Ideally, the running time should match the running time for detection when the output size is small and scale up to $\Theta(n^3)$ when the number of triangles reaches its maximum of $\Omega(n^3)$ (see e.g.,~\cite{bjorklund2014listing}). The problem of combinatorial Triangle Listing for general graphs is straightforward, as the trivial $\Order(n^3)$ time algorithm for detection can also list all the triangles in the same running time and is tight under \cref{conj:bmm}. In the setting of $H$-free graphs, the situation is less clear. Our lower bounds for detection in \cref{thm:lowerbounds} clearly hold for listing as well. For $H$ that does not fall into a lower bound class, the trivial $\Order(n^3)$ time listing algorithm seems wasteful because the output size is at most $ex(K_3,H)$, which is strongly subcubic by~\cite[Proposition 2.1]{alon2016many} (as $H$ is $3$-colorable). 

A general upper bound on $ex(K_3,H)$ that applies to every $3$-colorable $H$ is given by plugging a KST-type theorem for $3$-uniform hypergraphs~\cite[Theorem 1.5.7]{zhao2023graph} into the proof of~\cite[Proposition 2.1]{alon2016many}, resulting in $ex(K_3,H) = \Order(n^{3-1/k^2})$. For the special case where $H = C_{2k+1}$, it is established that $ex(K_3,C_{2k+1}) = \Order(ex(K_2,C_{2k}))$~\cite{gyHori2012maximum}, where $ex(K_2,C_{2k})$ is the maximum number of edges in a $C_{2k}$-free graph, bounded by $\Order(n^{1+1/k})$~\cite{bondy1974cycles}. A general lower bound by Kostochka et al.~\cite{kostochka2015turan} shows that for every $3$-colorable $H$, $ex(K_3,H) = \Omega(n^{3-\frac{3k-9}{e(H) - 3}})$, where $e(H)$ is the number of edges in $H$. In particular, for dense $3$-colorable patterns we get $ex(K_3,H) = n^{3-\Order(1/{k})}$. This gives a lower bound on the running time for listing, as the output size can be as large as $ex(K_3,H)$. 

Our Triangle Detection algorithms for $H$-free graphs generalize well to the listing setting, producing a list of all triangles in the same time it takes to detect a single one. The running time is independent of the output size because, as discussed, the output size in $H$-free graphs is bounded by $\Order(n^{3-1/k^2})$ (or $\Order(n^{1+1/k})$ if $H = C_{2k+1}$), whereas our algorithms run slower.

\begin{theorem}[Listing for nicely colored patterns]
\label{thm:listing-nice}
Let $H$ be a pattern that admits a nice $3$-coloring. Then there is a combinatorial algorithm that lists all triangles in an $H$-free graph, in time:
    \begin{equation*}
        \begin{cases}
            \tOrder(n^{3 - \frac{1}{2^{k-3}}}) & \text{ if $H$ is triangle-free,} \\
            \tOrder(n^{3 - \frac{1}{2^{k-4}}}) & \text{ if $H$ contains a triangle.}
        \end{cases}
    \end{equation*}
\end{theorem}

\begin{theorem}[$H$-sensitive listing for nicely colored patterns]
\label{thm:listing-sensitive}
For every pattern $H$ that admits a nice coloring, there is a combinatorial algorithm that, given a graph $G$ and a value $t$ which upper bounds the number of copies of $H$ in $G$, lists all the triangles in the graph. The running time is given by:
\[
    \begin{cases}
        \tOrder\left( n^{3-\frac{1}{2^{k-1}}} + n^{3 - \frac{k}{2^k - 1}} \cdot t^{\frac{1}{2^k - 1}}\right)  & \text{if } H \text{ is triangle-free}, \\
        \tOrder\left(n^{3-\frac{1}{2^{k-4}}} + n^{3-\frac{k-3}{2^{k-3}-1}} \cdot t^{\frac{1}{2^{k-3}-1}}\right) & \text{if } H \text{ contains a triangle}.
    \end{cases}
\]
\end{theorem}

\begin{theorem}[Listing for patterns with an attached triangle]
    \label{thm:listing-attached}
    Suppose that $H$ satisfies the following property. There is a $3$-coloring of $H$, such that after removing $s$ vertices with monochromatic neighborhoods until no such vertices remain, the leftover pattern consists of a triangle $xyz$ and a subpattern $H'$, with $xy \in E(H')$ and $z \notin V(H')$, and $H'$ is nicely colored. 
    Then there is a combinatorial algorithm that lists all triangles in an $H$-free graph in time $\tOrder\left(n^{3-\frac{1}{2^{k + s - 1}}}\right)$. The algorithm succeeds with probability $1-\Order(1/n)$. 
\end{theorem}

\begin{restatable}[Listing for odd cycles]{theorem}{listingodd}
\label{thm:listing-odd}
For every $k \geq 2$, there is a combinatorial algorithm that lists all triangles in a $C_{2k+1}$-free graph in $\tOrder( m + n^{1+2/k})$ time. 
\end{restatable}

In the case of nicely colored patterns, we observe a significant gap between $ex(K_3,H) = \Order(n^{3-1/k^2})$ and the stated running time. It is an interesting question whether the exponential dependence in our running time could be replaced by, e.g., a polynomial dependence $\Order(n^{3-1/k^{O(1)}})$. In the case of odd cycles, we see a smaller gap, a factor roughly $n^{1/k}$ between the combinatorial bound and the algorithmic running time. It seems likely that some algorithmic overhead is necessary for finding the triangles, even in $C_{2k+1}$-free graphs.

\paragraph{Paper organization.}
The paper is organized as follows. We begin with a high-level presentation of our core ideas in \cref{sec:overview}, followed by a discussion of related work in \cref{sec:related}. We then present the necessary preliminaries for our formal proofs in \cref{sec:preliminaries}. Our main algorithm for nicely colored patterns is presented in \cref{sec:H-nice}, proving \cref{thm:H-nice}. In \cref{sec:sensitive}, we extend this result to an $H$-sensitive algorithm and also present our specialized $C_5$-sensitive algorithm, proving \cref{thm:H-sensitive} and \cref{thm:C5-sensitive}. We then present our specialized algorithm for $C_{2k+1}$-free graphs in \cref{sec:odd}, proving \cref{thm:odd}. The proofs of \cref{prop:3k} and \cref{prop:smallpatterns} are given in \cref{sec:propositions}. Finally, we conclude with a summary and open questions in \cref{sec:conclusions}.

%% file: sections/overview.tex
In this section, we present the core ideas of the algorithms in this paper. We begin with the main result, which is the embedding algorithm.

\subsection{The embedding approach}
We begin with a simple $\tOrder(n^{7/3})$-time algorithm for $C_5$-free graphs, demonstrating the techniques and proving \cref{prop:C5-warmup}. 
\begin{proof}[Proof of \cref{prop:C5-warmup}]
    Let $1 \leq \Delta \leq n$ be a degree threshold to be set later. For every vertex with a degree smaller than $\Delta$, we can check if it participates in a triangle in $\Order(\Delta^2)$ time by examining all pairs of its neighbors. The algorithm begins by removing all vertices of degree less than $\Delta$ in a total of $\Order(n \Delta^2)$ time, until the remaining graph has a minimum degree of at least $\Delta$. Once the graph has a high minimum degree, a random sample of $\Theta(\frac{n}{\Delta} \log n)$ vertices will hit the neighborhood of a triangle vertex (if one exists) with high probability. We take such a sample and check whether any sampled vertex is itself a triangle vertex. This takes $\Order(n^2)$ per vertex for a total of $\tOrder(\frac{n^3}{\Delta})$ time. If no triangle is found, we proceed by processing the radius-$2$ balls around each of the sampled vertices, for a total of $\tOrder(\frac{n^3}{\Delta})$ time (as the number of edges in each ball could be as high as $\Theta(n^2)$). Observe that one of those balls contains the triangle with high probability. Now suppose that $v$ is a sampled vertex whose ball contains a triangle $xyz$, and let $N_1(v)$ and $N_2(v)$ denote the vertices at distances $1$ and $2$ from $v$, respectively. Since $v$ is not in a triangle (guaranteed above), there must be exactly one triangle vertex, say $x$, in $N_1(v)$, while the other two vertices, $y$ and $z$, are in $N_2(v)$. The approach is to search for such a triangle by examining every edge in $N_2(v)$ and a vertex in $N_1(v)$. Naively, this could take $\Order(n^3)$ time for each ball. 

    To exploit the $C_5$-free property, observe that for any edge $uw$ within $N_2(v)$, if the ``parents'' of its endpoints (i.e., neighbors of $u$ and $w$ in the upper layer $N_1(v)$) $p_w, p_u \in N_1(v)$ are distinct, a $C_5$ is formed with $v$. Since the graph is $C_5$-free, any two parents of $u$ and $w$ must be identical. This implies that $u$ and $w$ share a common parent, say $p \in N_1(v)$, which forms a triangle $puw$. Therefore, a triangle exists if and only if there is an edge within $N_2(v)$. This can be verified in $\Order(n^2)$ time per ball, for a total of $\tOrder(n^3/\Delta)$ time. The overall running time is $\tOrder(n \Delta^2 + n^3/\Delta)$, which is optimized to $\tOrder(n^{7/3})$ by setting $\Delta = n^{2/3}$.
\end{proof}

This algorithm for $C_5$-free graphs works by attempting to embed a $C_5$ in a $C_5$-free graph $G$. It gains a speed-up by restricting the search to a region of the graph where: (1) a triangle edge is highly likely to exist, and (2) any edge not part of a triangle would necessarily form a $C_5$. The goal is to generalize this approach to work for general patterns, but it is unclear what other patterns can be embedded. Hence, we will provide a general framework for this embedding approach, which will guide us in characterizing embeddable patterns. As mentioned earlier, we identify a broad class of patterns amenable to this embedding approach: those that admit a nice coloring, whose definition is restated below. 
\nice*

Suppose that a graph $G$ contains a triangle $xyz$, and suppose for simplicity that it is the only triangle in the graph (in \cref{sec:lowerbounds} we will see that we can assume so without loss of generality).
An embedding algorithm takes a pattern $H$ and embeds its vertices one at a time in $G$. At stage $i$, for $i = 1,2,\ldots,k$, a vertex $h_i \in V(H)$ is mapped to a vertex $v_i \in V(G)$. 
The basic requirements for an embedding algorithm are as follows:
\begin{itemize}
    \item For each $i \in [k]$, there exists some $t_i \in \set{x,y,z}$ such that $v_i \in N(t_i)$. This ensures that the area surrounding the embedded pattern contains the triangle. 
    \item For each $i \in [k]$, $v_i$ lies in the common neighborhood of all the previously embedded neighbors of $h_i$. Formally, 
    \[
        \forall i \in [k]: \quad v_i \in \bigcap\limits_{\substack{j<i \text{ such that}\\ h_i h_j \in E(H)}} N(v_j).
    \]
    This ensures that the embedding respects the edges of $H$. 
\end{itemize}
An important observation is that for every $i \in [k]$, $t_i$ must be unique, i.e., $v_i$ can only be adjacent to a single vertex in $\set{x,y,z}$. That is because otherwise, there would be another triangle in $G$ containing $v_i$ and two vertices in $\set{x,y,z}$. Moreover, the following requirement from the algorithm follows: For every neighbor $h_j$ of $h_i$, we must have $t_i \neq t_j$. That is because $v_j$ and $v_i$ are adjacent (by the second requirement), so if $t_i = t_j$, there would be another triangle $v_i v_j t_i$ in $G$. This implies that the patterns that the algorithm can embed must be $3$-colorable: let the color of $h_i$ be $t_i$ for every $i \in [k]$.

Note that there is no guarantee that the intersection set from the second bullet above is non-empty. If it becomes empty at some stage, then the embedding algorithm gets stuck. 
To have this guarantee, we add another requirement from the embedding algorithm. The requirement is that for each $i \in [k]$ and every two neighbors of $h_i$, $h_j$ and $h_\ell$, where $j,\ell < i$, the triangle neighbors of $v_j$ and $v_\ell$ will be the same. Formally:
\[
    \forall i \in [k]: \quad \forall \substack{j,\ell < i \\ h_j, h_\ell \in N_H(h_i)}: \quad t_j = t_\ell.
\]
This ensures that the set $N(v_j) \cap N(v_\ell)$ is non-empty, as it contains a triangle vertex $t_j = t_\ell$. In particular, the intersection set from the second bullet is non-empty because it contains a unique triangle vertex. Note that the $3$-coloring of $H$ we now get is a nice coloring: Starting from the last vertex in the embedding, every vertex has a monochromatic neighborhood in $H$. 

\paragraph{An algorithm for nicely colored patterns}
We are now ready to present the embedding algorithm from \cref{thm:H-nice} for patterns that are nicely colored. 
Suppose that $H$ is nicely colored and that $G$ is $3$-colored (in \cref{sec:lowerbounds} we will see that this assumption is without loss of generality), so in particular, $G$ is free from nicely colored copies of $H$. The algorithm maintains the invariant of a high min-degree to every color class. Namely, suppose that some vertex (e.g., a red vertex) contains fewer than $\Delta$ neighbors in some other color class (e.g., the blue vertices). In that case, it can be verified in $\Order(n \Delta)$ time that the vertex does not participate in a triangle, resulting in a total of $\Order(n^2 \Delta)$ time for a low-degree cleanup. In each stage, the vertex of $H$ whose neighborhood is monochromatic is embedded by taking a random sample of $\tOrder(n/\Delta)$ vertices from its color class in $G$. Then, the color class of $G$ corresponding to its monochromatic neighborhood is restricted to the neighborhoods of the sampled vertices (creating $\Order(n/\Delta)$ different subgraphs of $G$). Because of high minimum degrees, at least one sampled vertex is guaranteed to hit the neighborhood of a triangle vertex, so the resulting subgraph still contains the triangle. As the process continues, the color classes of $G$ are restricted to neighborhoods of embedded vertices. The triangle is found when it becomes incident to a low-degree vertex, which will surely happen once a base case is reached. 
We remark that after each stage, a new low-degree cleanup needs to be applied to every subgraph. Applying this cleanup naively to each of the $\tOrder(n/\Delta)$ new subgraphs would take cubic $\tOrder( \frac{n}{\Delta} \cdot n^2 \Delta)$ time. Instead, a decreasing sequence of thresholds $\Delta_k > \Delta_{k-1} > \cdots > \Delta_1$ is introduced, and a different threshold is used in each recursive level. The final running time is obtained by optimizing them.  

\paragraph{The case of an attached triangle.} Let us now explain the algorithm from \cref{thm:attached}. The idea is to first embed the vertices with monochromatic neighborhoods, as in the case of nicely colored patterns. Then, we reach the hard-core of the pattern: a nicely colored subpattern with an attached triangle on one of its edges. For this pattern, the key idea is to sparsify the graph by reducing the number of copies of the subpattern $H'$. Namely, every edge that participates in many colored copies of $H'$ is removed. Once the number of copies of $H'$ becomes polynomially smaller than $n^{|V(H')|}$, we apply an $H'$-sensitive algorithm as a subroutine to the remaining graphs.

\subsection{The approach for $H$-sensitive algorithms.}
It is possible to generalize the embedding algorithm to be $H$-sensitive by employing the following idea: In a sample of $\tOrder(n/\Delta)$ vertices in a graph with high-minimum degree and at most $t$ copies of $H$, one of the sampled vertices is an \emph{$H$-light} neighbor of the triangle with high probability. That is, a neighbor of the triangle that participates in at most $\Order(t/\Delta)$ copies of $H$. This ensures that the number of copies of the forbidden pattern reduces by a factor of $\Delta$ at each step. When a base case is reached, the number of copies of a very basic pattern (e.g., a single vertex) is small enough for the base case to be solved efficiently. For instance, if the number of vertices of a certain color is small, triangles incident to them can be found quickly.

\paragraph{A specialized $C_5$-sensitive algorithm.}
We now explain how to improve the running time of the previously discussed $\tOrder(n^{7/3})$-time algorithm for $C_5$-free graphs to $\tOrder(n^2)$ by changing a single component: Instead of cleaning up low-degree vertices and then focusing on the remaining graph that has high min-degree, we guess the maximum degree of a vertex in a triangle, and then ignore vertices with higher degrees in $G$. The guess can be implemented by bucketing, i.e., trying a sequence of exponentially increasing guesses. We then follow the same steps as in the $\Order(n^{7/3})$ time algorithm: Sample $\tOrder(n/\Delta)$ vertices and search for triangles in the $2$-hop ball around each sampled vertex. However, instead of constructing a $2$-hop ball around each vertex, we only process vertices of degree at most $\Delta$. The time to process each ball becomes $\Order(n \Delta)$, for a total of $\tOrder( \frac{n}{\Delta} \cdot n \Delta ) = \tOrder(n^2)$ time. To obtain a $C_5$-sensitive algorithm, we relate the search time for a triangle inside a random $2$-hop ball to the number of $5$-cycles in that ball. Using the fact that the expected number of $5$-cycles incident to a random vertex is $\Order(t/n)$, we get the bound of $\tOrder(n^2 + n^{4/3} t^{1/3})$.

\subsection{The approach for $C_{2k+1}$-free graphs.}
Let us now explain the high-level overview of the algorithm for $C_{2k+1}$-free graphs. Unlike the previous algorithms, this one is deterministic and employs a significantly different approach. The idea is to process the layers of a $k$-hop ball around a vertex as long as the layers are expanding. The algorithm searches for triangles in the first $k-1$ layers using a generalized Chiba-Nishizeki algorithm. If a triangle is not found, then most of the ball is deleted (except for the last two layers, which may still contain a triangle), ensuring a bounded number of phases before we run out of vertices. This approach exploits the structure of $C_{2k+1}$-free graphs using three main components:
\begin{itemize}
    \item A lemma stating that the subgraph induced by any single layer in the ball is $O(k)$-degenerate.
    \item A generalization of the Chiba-Nishizeki algorithm that can find triangles with at least one edge in a specified $O(k)$-degenerate subgraph. This algorithm is useful for finding a triangle with an edge that lies within a layer.
    \item An efficient algorithm that identifies a large fraction of edges between the last two layers of the ball that participate in a $2k$-cycle, and for each such edge, it finds a $2k$-cycle containing it.
\end{itemize}
The last component is crucial: An edge that participates in a $2k$-cycle cannot form a triangle with a vertex outside that cycle without creating a $(2k+1)$-cycle. Thus, to check such an edge for triangles, one only needs to inspect the adjacencies to the other $2k-2$ vertices of a known $2k$-cycle containing it, which takes $\Order(k) = \Order(1)$ time.

%% file: sections/related.tex
\subsection{Complexity dichotomies in $H$-free graphs}
A large body of research is devoted to proving \emph{dichotomy theorems} for $H$-free graphs: showing a separation between patterns for which a problem is ``hard'' in $H$-free graphs and patterns for which the problem becomes ``easy''. The precise definitions of hard and easy may vary, but most of this research uses coarse notions of hardness, such as NP-completeness versus polynomial-time. For example, such theorems have been proven for Max-Cut, Maximum Independent Set, and List-Coloring~\cite{kaminski2012max,alekseev1992complexity,golovach2012coloring}. An almost complete dichotomy also exists for Subgraph Isomorphism~\cite{bodlaender2020subgraph}.

Interestingly, a comprehensive work by Johnson et al.~\cite{johnson2025complexity} provided a meta-theorem that captures many problems for which such a dichotomy exists, using a framework applicable in different complexity regimes. Roughly speaking, the meta-theorem states that if a problem (e.g., Maximum Independent Set) is easy on bounded-treewidth graphs, hard on subcubic graphs (graphs of degree at most $3$), and hard on edge subdivisions of subcubic graphs, then a complete dichotomy exists between patterns for which the problem is easy and patterns for which it is hard in $H$-free graphs. Their framework is applied to various problems, primarily in the coarse complexity regimes. However, the paper also demonstrates its generality by applying it to problems in $P$. Namely, they consider the Diameter and Radius problems, showing a separation between patterns for which the problem is solvable in near-linear time and patterns for which the problem does not admit a strongly subquadratic algorithm.

Our work uses the hard and easy notions of cubic and strongly subcubic time, respectively. Under these notions, Triangle Detection does not fall into their framework because the problem is solvable in linear time on subcubic graphs, which makes it easy. One may wonder about the status of the NP-complete and parameterized generalizations of Triangle Detection, i.e., Maximum Clique and $k$-Clique. For Maximum Clique, the problem in $H$-free graphs is trivially solvable in polynomial time for any fixed $H$, as the largest clique in an $H$-free graph has size less than $|H|$. For $k$-Clique, one could seek an FPT/W[1] dichotomy between patterns. This problem is interesting only when $|H| > k$, but it appears not to have been considered in the literature yet.

\paragraph{The induced case.} A similar and equally important question, which we do not address in this work, concerns the complexity of Triangle Detection in \emph{induced} $H$-free graphs. Every $H$-free graph is also an induced $H$-free graph, but the opposite does not necessarily hold. This implies that dichotomies in the two settings differ; for example, a pattern $H$ could exist for which the problem is hard on induced $H$-free graphs but easy in $H$-free graphs. The induced setting also behaves more nicely with respect to graph inversion because a graph $G$ is induced $H$-free if and only if its complement $\bar{G}$ is $\bar{H}$-free (a similar behavior does not hold in the $H$-free setting). Inverting the graph is useful when considering the generalizations of Triangle Detection, like $k$-Clique and Maximum Clique, because of the duality between these problems and their Independent Set counterparts. The Maximum Independent Set problem on induced $H$-free graphs is very well-studied, and a recent work established a dichotomy for Maximum Independent Set between patterns for which the problem is NP-complete and those for which it is in quasipolynomial time~\cite{gartland2024maximum}. Whether this extends to a P/NP dichotomy remains a major open question. In the parameterized setting, one paper has made progress toward an FPT/W[1] dichotomy for the $k$-Independent-Set problem~\cite{bonnet2020parameterized}.

\subsection{Distance oracles and short cycle removal}
\paragraph{The approximate distance oracle approach.} An alternative approach for graphs that are $C_{2\ell+1}$-free for every $\ell \in [2, k]$ is to use approximate distance oracles: a data structure that upon a query $(u,v)$ returns an approximate distance $\mathrm{dist}(u,v) \leq \tilde{d} \leq t \cdot \mathrm{dist}(u,v)$ for some \emph{stretch} $t \geq 1$. Suppose that $G$ is a tripartite graph with parts $A, B, C$, and we delete all edges in $E(B,C)$. A distance query for a deleted edge $uv \in E(B,C)$ can determine whether $uv$ participates in a triangle in $G$: $\mathrm{dist}(u,v) = 2$ if and only if $uv$ is part of a triangle. If $G$ is $C_{2\ell+1}$-free for all $\ell \in [2, k]$, then for non-triangles, $\mathrm{dist}(u,v)$ is an odd number strictly greater than $2k+1$. Hence, an approximate distance oracle with an appropriate stretch can distinguish between the YES and NO cases. For example, the Thorup-Zwick oracle~\cite{thorup2005approximate} with a stretch $t$ given by:
\[
t =
\begin{cases}
k/2 & k \text{ even} \\
(k+1)/2 & k \text{ odd}
\end{cases}
\]
distinguishes between these cases and runs in time $\Order(m n^{1/t})$. This improves upon our algorithm from \cref{thm:odd} only for odd $k$ and for graphs with $m = \Order(n^{1+\frac{2}{k(k+1)}})$ edges. Moreover, this technique does not address graphs that are only $C_{2k+1}$-free, whereas our approach in \cref{thm:odd} does.

\paragraph{Short cycle removal.} A paper by Abboud et al.~\cite{ABKZ22} shows that for $\Theta(\sqrt{n})$-regular graphs, results opposite to ours hold. Roughly speaking, they show a self-reduction that reduces the number of $k$-cycles from a maximum of $n^{k/2+1/2}$ to $\Order(n^{k/2+\gamma})$ for some $\gamma < 1/2$, without removing triangles. Furthermore, they show a self-reduction that removes all copies of $C_4$ from a graph with maximum degree $\Order(\sqrt{n})$ without removing any of its triangles. This powerful technique provides lower bounds for approximate distance oracles, $4$-Cycle Detection, and Triangle Detection in $C_4$-free graphs. Subsequent works~\cite{abboud2023stronger,jin2023removing} improved these results via $3$-SUM reductions that output worst-case graphs with fewer than the maximum number of $k$-cycles. This reveals a surprising dichotomy: While our algorithm from \cref{thm:H-sensitive} gains a significant speed-up on dense graphs with fewer than the maximum number of copies of a pattern $H$, in the sparse regime, worst-case instances do not necessarily have the maximum number of copies of the pattern. In fact, the result for $4$-cycles shows that worst-case instances of Triangle Detection can even have fewer than the \emph{average} number of $C_4$'s in a random $\sqrt{n}$-regular graph.

The current best lower bound for the running time of approximate distance oracles is from~\cite{abboud2023stronger}, but the stretch they obtain is still a factor of $2$ shy of the stretch from the Thorup-Zwick oracle for the same running time. One might try to improve upon the self-reduction technique of~\cite{ABKZ22}, for example, by picking different density regimes or by attempting to remove only short odd cycles instead of all short cycles. The algorithm from \cref{thm:odd} places a limit on what can be proven with such techniques, as it runs faster than the Thorup-Zwick oracle for almost every $m$ and $k$.

\subsection{Additional works}
\paragraph{The paired pattern problem} 
Dalirrooyfard et al.~\cite{dalirrooyfard2022induced} considered a paired pattern problem that asks to detect whether a graph contains at least one of two patterns. For example, the $\set{K_3,H}$-Detection problem asks to detect whether a graph contains a triangle or a copy of $H$. An algorithm for $\set{K_3,H}$-Detection clearly solves the Triangle Detection problem in $H$-free graphs. However, the opposite does not necessarily hold. Still, one may ask whether our embedding approach solves the paired pattern problem as well. The answer is negative, because our algorithm relies on deleting low-degree vertices after it verifies that they do not participate in a triangle. However, those vertices may participate in copies of $H$, and the algorithm will fail to detect them. 

\paragraph{Listing large cliques from a few small cliques.} Another result by Dalirrooyfard et al.~\cite{dalirrooyfard2024towards} for $k$-Clique Listing is an algorithm whose running time is sensitive to the number of $\ell$-cliques for $\ell < k$. This result resembles our $H$-sensitive algorithms but in the opposite direction: the pattern with few copies is a subpattern of the target pattern. In contrast, our work considers patterns that are larger than a triangle and do not necessarily contain one.

\paragraph{Triangle Detection in dense graph families.} Several papers have considered Triangle Detection in various graph families, resulting in very efficient combinatorial algorithms. The only such families that are also dense that we are aware of are geometric intersection graphs. A work by Kaplan et al.~\cite{kaplan2019triangles} shows a near-linear time algorithm for disk graphs, and a recent work by Chan~\cite{chan2023finding} shows strongly subcubic algorithms for intersection graphs of more objects: boxes, line segments, fat objects, and translates. While these algorithms use fast matrix multiplication, they remain subcubic even when using naive matrix multiplication.

%% file: sections/preliminaries.tex
We begin by establishing the notation and definitions used throughout the paper. For integers $a \leq b$, we let $[a,b] = \set{a,a+1,\ldots,b}$. For $k \in \Nat$, we let $[k] = [1,k]$.

\paragraph{Graph notation.}
Throughout, let $G=(V,E)$ be a graph with $n=|V(G)|$ vertices and $m=e(G)=|E(G)|$ edges. 
For a vertex $v \in V$, its neighborhood is denoted by $N(v)$, and the set of vertices at distance $2$ from $v$ is denoted by $N_2(v)$. An edge between vertices $u$ and $v$ is denoted by $uv$, and a triangle on vertices $x,y,z$ is denoted by $xyz$. For a subset of vertices $S \subseteq V$, $G[S]$ denotes the subgraph induced by $S$.  For a set $S \subseteq V(G)$ we define $N_G(S) = \bigcup_{v \in S} N_G(v)$. We omit the subscript when it is clear from the context. Given two vertex sets $X, Y \subseteq V(G)$, we denote the set of edges with one endpoint in $X$ and one in $Y$ by $E(X,Y)$, and its cardinality by $e(X,Y)$. For a rooted tree $T$ and a vertex $v \in V(T)$, we define $T_v$ to be the subtree rooted at $v$.

A graph is said to be \emph{properly $3$-colored} if each vertex is assigned one of the three colors in $\set{\texttt{RED},\texttt{BLUE},\texttt{GREEN}}$, such that no two adjacent vertices share the same color. For brevity, we will simply refer to such graphs as being $3$-colored.

\paragraph{Subgraph copies.} We now formally define what it means for a graph $G$ to contain a smaller pattern graph $H$.

\begin{definition}[Copy of $H$ in $G$]
    \label{def:copies}
    A \emph{copy of $H$ in $G$} is an injective homomorphism $\phi \colon V(H) \to V(G)$; that is, $\phi$ is one-to-one and for every edge $xy \in E(H)$, we have $\phi(x)\phi(y) \in E(G)$. The number of copies of $H$ in $G$ is the number of such injective homomorphisms. We say that $G$ is \emph{$H$-free} if it contains no copies of $H$.
\end{definition}

\begin{definition}[Colored copy of $H$ in $G$]
    \label{def:colored-copies}
    Given a $3$-colored graph $G$ with coloring $C_G \colon V(G) \to \{\texttt{RED}, \texttt{BLUE}, \texttt{GREEN}\}$, and a $3$-colored pattern $H$ with coloring $C_H \colon V(H) \to \{\texttt{RED}, \texttt{BLUE}, \texttt{GREEN}\}$, a \emph{colored copy of $H$} is an injective homomorphism $\phi \colon V(H) \to V(G)$ that preserves colors, i.e., $C_G(\phi(v)) = C_H(v)$ for every $v \in V(H)$. The terms ``number of colored copies'' and ``colored $H$-free'' are defined analogously.
\end{definition}

Note that a colored $H$-free graph may contain many copies of $H$, but none that respect the specified coloring of $H$.

\paragraph{Graph cores.} A graph is \emph{$\Delta$-degenerate} if every one of its induced subgraphs has a vertex of degree less than $\Delta$. Related to degeneracy is the notion of a graph core.

\begin{definition}[$\Delta$-core]
    Given a graph $G$, the \emph{$\Delta$-core of $G$} is the maximal induced subgraph of $G$ with a minimum degree of at least $\Delta$.
\end{definition}

\begin{claim}
    \label{claim:core}
    For any integer $\Delta \geq 1$, the $\Delta$-core of a graph can be computed in linear time.
\end{claim}

\begin{proof}
The $\Delta$-core can be found by repeatedly deleting vertices of degree less than $\Delta$ until no such vertices remain. To do this in linear time, we can partition the vertices into $n$ lists, $D_0, D_1, \ldots, D_{n-1}$, according to their degrees. As long as any list $D_i$ with $i < \Delta$ is non-empty, we remove an arbitrary vertex $v$ from it, delete $v$ from the graph, and update the degrees of its neighbors, relocating them to the lists corresponding to their new degrees. The total time for this procedure is $\Order(n+m)$.
\end{proof}

We use an analogous notion for 3-colored graphs.

\begin{definition}[Colored $\Delta$-core]
    The \emph{colored $\Delta$-core} of a $3$-colored graph $G$ is the maximal induced subgraph where every vertex has at least $\Delta$ neighbors in each of the two color classes other than its own.
\end{definition}

The colored $\Delta$-core can be computed in linear time by a similar iterative removal process.

\paragraph{Sum of exponents.} Finally, we will make use of the following standard identity for the sum of a geometric series.
\begin{equation}
    \label{eq:exponents}
    \forall j \leq k :\quad \sum_{i=j}^k \frac{1}{2^{i}} = \frac{1}{2^{j-1}} - \frac{1}{2^k}.
\end{equation}

%% file: sections/lowerbounds.tex
In this section, we prove \cref{thm:lowerbounds}. 
\lowerbounds*
The proof utilizes two well-known self-reduction techniques, whose short proofs are included for completeness. The first part of the theorem is a direct consequence of the color-coding technique~\cite{alon1995color}.

\begin{lemma}[Color-coding]
\label{lem:color-coding}
There is a near-linear time reduction that, given a graph $G$ with $n$ vertices, produces $\ell = \Order(\log n)$ subgraphs $G_{1},G_2,\ldots,G_{\ell} \subseteq G$, and $\ell$ proper $3$-colorings $c_1,\ldots,c_\ell$, $c_i \colon V(G_i) \mapsto \set{\texttt{RED}, \texttt{GREEN}, \texttt{BLUE}}$, such that if $G$ contains a triangle then with probability $1-\Order(1/n)$, at least one subgraph in the list contains a triangle. 
\end{lemma}
\begin{proof}
Take a random coloring $c \colon V(G) \mapsto \set{\texttt{RED}, \texttt{GREEN}, \texttt{BLUE}}$, where the color of each vertex is chosen uniformly and independently. Then remove from $G$ every monochromatic edge, resulting in a properly $3$-colored subgraph $G'$. If $G$ contains a triangle, then the probability that a triangle survives in $G'$ is at least $2/9$. By repeating this process $\ell = \Theta(\log n )$ times, we get $\ell$ different $3$-colored subgraphs. The probability that at least one of them contains a triangle is at least: 
\[
1- (7/9)^{\ell} = 1-n^{-\Theta(1)}.
\]
By choosing an appropriate constant factor in $\ell$, we get the bound $1-\Order(1/n)$. 
\end{proof}

The second part of the theorem is a consequence of the so-called sieving technique (see e.g.,~\cite{dell2022approximately,goldreich2025coarse}). The idea is to produce a subgraph of $G$ that contains at most one triangle.

\begin{lemma}[Sieving]
\label{lem:sieving}
There is a near-linear time reduction that, given a $3$-colored graph $G$ with at most $n$ vertices in each part, outputs $\Order(\log^3(n))$ subgraphs of $G$, such that if $G$ contains a triangle, then with probability at least $1-\Order(1/n)$, one of the subgraphs contains a unique triangle.
\end{lemma}

\begin{proof}
The reduction is as follows. For every $p \in \set{1, 1/2, 1/4, \dots, 2^{-\lceil \log n \rceil}}$, subsample the red vertices with probability $p$ (i.e., each red vertex survives independently with probability $p$) and denote the resulting subgraph by $G_p$. Now, for each such subgraph $G_p$ and for every $q \in \set{1, 1/2, 1/4, \dots, 2^{-\lceil \log(n^2) \rceil}}$, subsample the edges between green and blue vertices with probability $q$, denoting the resulting subgraph by $G_{p,q}$. The total number of subgraphs is $\Order(\log^2 n)$. We now claim that if $G$ contains a triangle, then with constant probability, one of these subgraphs contains precisely one triangle.

Suppose there are $1 \leq r \leq n$ red vertices that participate in a triangle. Let $p = 2^{-\lceil \log r \rceil}$. The probability that precisely one red vertex that participates in a triangle survives in $G_p$ is $r p (1-p)^{r-1}$. By analyzing the derivative of the function $f(x) = rx(1-x)^{r-1}$ over the interval $[1/(2r), 1/r]$, we see that it is non-decreasing, so its minimum is obtained at $x=1/(2r)$. Thus, since $1/(2r) \leq p$, we have $r p (1-p)^{r-1} \geq r \cdot \frac{1}{2r} (1-\frac{1}{2r})^{r-1} = \frac{1}{2} (1-\frac{1}{2r})^{r-1} \geq 1/(2\sqrt{e})$.

Now, suppose $G_p$ is fixed and contains a single red vertex, say $v$, that participates in a triangle (possibly more than one). Let $1 \leq x \leq n^2$ be the number of blue-green edges that form a triangle with $v$. By the same argument, for $q = 2^{-\lceil \log x \rceil}$, the probability that exactly one such edge survives in $G_{p,q}$ is at least $1/(2 \sqrt{e})$. 
Hence, if $G$ contains a triangle, then with probability at least $1/(4e)$, there exists some subgraph $G_{p,q}$ that contains a single triangle. 
To boost the success probability, repeat this entire process $\Theta(\log n)$ times. This yields a total of $\Order(\log^3 n)$ subgraphs, and the failure probability becomes polynomially small, e.g., $\Order(1/n)$.
\end{proof}

\begin{proof}[Proof of \cref{thm:lowerbounds}]
Suppose first that $H$ is not $3$-colorable and that there is a combinatorial Triangle Detection algorithm for $H$-free graphs running in $n^{3-\eps}$ time for some $\eps > 0$, and succeeding with probability $1-\Order(1/n^2)$. Given a Triangle Detection instance $G$, we apply the reduction from \cref{lem:color-coding} to obtain a list of $\ell = \Order(\log n)$ $3$-colored subgraphs of $G$. Every subgraph is clearly $H$-free because $H$ is not $3$-colorable. Applying the supposed algorithm to every subgraph, we get $\ell$ answers, and with probability at least $1-\Order(\ell/n^2) = 1-\Order(1/n)$, they are all correct. If $G$ contains a triangle, then by the lemma, one of those subgraphs contains a triangle with probability $1-\Order(1/n)$, so in that case, our algorithm detects a triangle with probability $1-\Order(1/n)$. If $G$ does not contain a triangle, then none of the subgraphs contains a triangle, and therefore the algorithm does not report a triangle. We showed a combinatorial Triangle Detection algorithm running in $\Order(n^{3-\eps} \cdot \log n)$ time, thereby refuting \cref{conj:bmm}.

Now, suppose that $H$ contains more than one triangle and there is a combinatorial algorithm for Triangle Detection on $H$-free graphs running in $n^{3-\eps}$ time, for some $\eps > 0$, and succeeding with probability $1-\Order(1/n^2)$. Given a Triangle Detection instance $G$, first apply the reduction from \cref{lem:color-coding} to get a list of $\ell = \Order(\log n )$ $3$-colored subgraphs. Then, apply the reduction from \cref{lem:sieving} on every subgraph to get a total of $\Order(\log n \cdot \log^3(n)) = \Order( \log^4(n) )$ subgraphs. We execute the supposed algorithm on each of these subgraphs, and report a triangle if and only if the supposed algorithm reports a triangle on at least one of its inputs. If $G$ has a triangle, then by \cref{lem:sieving}, with probability $1-\Order(1/n)$ at least one of these subgraphs contains a single triangle, and therefore it is $H$-free. In that case, the supposed algorithm reports a triangle in that subgraph with probability $1-\Order(1/n)$, and a triangle is correctly reported by the algorithm for $G$. 
If $G$ is triangle-free, then none of the subgraphs contain a triangle and therefore they are all $H$-free. In particular, the supposed algorithm correctly answers NO 
on all those subgraphs with probability $1-\tOrder(\ell/n^2) = 1-\Order(1/n)$, so a triangle is not reported in $G$. We showed a combinatorial Triangle Detection algorithm running in $\Order(n^{3-\eps} \cdot \log^4(n))$ time, thereby refuting \cref{conj:bmm}.
\end{proof}

%% file: sections/H-nice.tex
In this section, we prove \cref{thm:H-nice} and its generalization to listing, \cref{thm:listing-nice}.
\algnice*

Our algorithm will solve a slight variant of the problem: Triangle Detection in \emph{colored} $H$-free graphs, where $H$ is nicely colored (recall \cref{def:colored-copies}). To solve the original uncolored problem, we apply the following preprocessing: given an $H$-free graph $G$, compute in linear time a $3$-colored subgraph $G' \subseteq G$, using \cref{lem:color-coding}. 
Then compute a nice coloring of $H$ in $\Order(3^k) = \Order(1)$ time. Since $G$ is $H$-free, $G'$ is colored $H$-free. We henceforth assume that $G$ is $3$-colored and $H$ is nicely colored. 

The following claim will be useful throughout the paper. 
\begin{claim}[Low-degree cleanup]
    \label{claim:lowdeg}
    There is an algorithm that, given a $3$-colored graph $G$ and a parameter $1 \leq \Delta \leq n$, decides whether $G$ contains a triangle that is incident to a vertex outside the colored $\Delta$-core of $G$. If no such triangle is found, it returns the colored $\Delta$-core of $G$. The algorithm runs in $\Order(n^2 \Delta)$ time.
\end{claim}
\begin{proof}
To find a triangle, we process every vertex that has at most $\Delta$ neighbors in some color class other than its own. For each such vertex, we examine every bichromatic pair of its neighbors to search for an edge. If an edge is found, we report a triangle containing that vertex and edge. Otherwise, we delete the vertex from the graph. If no triangles are found, the remaining graph is the colored $\Delta$-core of $G$.

For the analysis, there is a cost of $\Order(n \Delta)$ per vertex to check for triangles, so a total of $\Order(n^2 \Delta)$ time for all those checks. It is also possible to implement a data structure that processes the vertices and supports deletions similarly to what is done in \cref{claim:core}. Hence, the total running time is $\Order(n^2 \Delta)$. 
\end{proof}

The algorithm will be recursive, and in each recursive step, the pattern is reduced by one vertex until a base case is reached. We define the base case to be a colored pattern of size $3$ if $H$ is triangle-free, or a colored pattern of size $4$ if $H$ contains a triangle. Let us now show that the base cases can be solved in $\Order(n^2)$ time. 
\begin{claim}[Base cases]
    \label{claim:base-case}
    Let $H$ be a nicely-colored pattern on three vertices, or a nicely-colored pattern on four vertices that contains a triangle. 
    Then, there is an algorithm that solves Triangle Detection on colored $H$-free graphs in $\Order(n^2)$ time. 
\end{claim}
\begin{proof}
    Let us begin with patterns on three vertices. Suppose that $G$ is colored $H$-free. 
    If the three vertices of $H$ have distinct colors, then $G$ must be triangle-free, because a colored triangle contains $H$ as a colored subgraph. This case can thus be solved trivially in constant time by answering NO. 
    Now suppose that $H$ is colored by two colors, so in particular it has at most two edges. 
    Every such pattern is a subgraph of a path $v_1 - v_2 - v_3$, where $v_1$ and $v_3$ have the same color, so it suffices to focus on this case only. Without loss of generality, suppose that $v_1$ and $v_3$ are colored red and $v_2$ is colored blue. Then, in $G$, every blue vertex has at most one red neighbor to avoid forming a colored copy of $H$. We invoke \cref{claim:lowdeg} with $\Delta = 1$ to solve this case in $\Order(n^2)$ time. 
    Another case is when $H$ is colored only by one color. This implies that $H$ is empty, and therefore, there is a color class in $G$ that has at most two vertices. We invoke \cref{claim:lowdeg} with $\Delta = 2$ to solve this case in $\Order(n^2)$ time. 

    Now consider the case where $H$ has four vertices and contains a triangle. We may assume without loss of generality that $H$ is a colored triangle with an attached edge, because the other patterns are either a subgraph of this or contain more than one triangle (thus not nicely colored). Suppose that the attached edge is attached to the blue vertex of the triangle, and its other endpoint is red. This implies that in $G$, the blue vertex of the triangle (if it exists) can have exactly one red neighbor (its neighbor in the triangle). We invoke \cref{claim:lowdeg} with $\Delta = 1$ to solve this case in $\Order(n^2)$ time.
 \end{proof}

The input to the recursive algorithm is a quadruple $(G, C_G, H, C_H)$, where $C_G$, $C_H$ are the colorings of $G$ and $H$ respectively. 
The algorithm uses a sequence of parameters $\Delta_k > \Delta_{k-1} > \cdots > \Delta_4 > \Delta_3 $, where $\Delta_i$ corresponds to a pattern of size $i$. We will set the values of these parameters at the end of the proof, except for the ones corresponding to the base cases, which are set to $1$ (e.g. $\Delta_4=1$ if $H$ contains a triangle, and otherwise $\Delta_3 = 1$). Throughout the recursion, for a given input $(G,C_G,H,C_H)$, we denote by $k$ the number of vertices in the input pattern $H$ and by $n$ the number of vertices in the input graph $G$. We use $V_c(G)$ to denote the set of vertices of color $c \in \set{\texttt{RED},\texttt{BLUE},\texttt{GREEN}}$ in $G$. 

\begin{proof}[Proof of \cref{thm:H-nice}]
At each recursive level, the following procedure is executed: 
\begin{itemize}
    \item \textbf{Base-case:} 
    If $k=3$ or if $k=4$ and $H$ contains a triangle, invoke the algorithm from \cref{claim:base-case} and return its output. 

    \item \textbf{Low-degree cleanup:} 
    Invoke the algorithm from \cref{claim:lowdeg} with $G$ and $\Delta_k$. If a triangle is found, return it. Otherwise, continue with the colored $\Delta_k$-core of $G$, denoted $G'$. 

    \item \textbf{Random embedding of $h$ in $G$:}
    Let $h \in V(H)$ be a vertex whose neighborhood is monochromatic (guaranteed by the assumption on $C_H$), and denote the color of its neighbors by $c^*$. Sample uniformly at random $\ell := \Theta( \frac{n}{\Delta_k} \cdot \log n)$ vertices $s_1, \dots, s_\ell$ from $V_{C_H(h)}(G')$. For each $i \in [\ell]$, create a subgraph $G_i \subseteq G'$ by deleting every vertex of color $c^*$ that is not a neighbor of $s_i$. Let $C_{G_i}$ be the restriction of $C_G$ to $G_i$.

    \item \textbf{Recursive calls:}  
    Let $H' = H[V(H) \setminus \{h\}]$ and let $C_{H'}$ be the restriction of $C_H$ to $H'$. Since the property of being a nice coloring is closed under vertex deletion, $C_{H'}$ is also nice. Make a recursive call on $(G_i, C_{G_i}, H', C_{H'})$ for each $i \in [\ell]$. If any of those calls returns a triangle, then return it. Otherwise, report that $G$ is triangle-free.
\end{itemize}

    \paragraph{Correctness.} Let us prove that if $G$ contains a triangle, then the algorithm finds one with probability $1 - \Order(1/n)$. 
    It is immediate that if $G$ is triangle-free, then the algorithm does not find a triangle in $G$.
    Now suppose $G$ contains a triangle $xyz$. If the triangle is not found in the low-degree cleanup phase, then each vertex of the triangle has at least $\Delta_k$ neighbors of color $C_H(h)$ (the vertex $h \in V(H)$ chosen for embedding). Suppose that the triangle vertex whose color is $c^*$ is $x$. For the triangle to survive to the next recursive level, one of the neighbors of $x$ whose color is $C_H(h)$ must be selected in the random sample of $\ell = \Theta\left(\frac{n}{\Delta_k} \cdot \log n\right)$ vertices. Since the number of such neighbors is at least $\Delta_k$, and the total number of vertices of color $C_H(h)$ is at most $n$, the probability that none are selected is at most:
    \[
        \left(1 - \frac{\Delta_k}{n} \right)^\ell \leq \exp\left( -\frac{\Delta_k}{n} \cdot \ell \right) = \frac{1}{n^{\Theta(1)}}
    \]
    So by choosing the constant factor of $\ell$ appropriately, the probability that the triangle survives is $1- \Order(1/n)$. 
    Now, we claim that the graphs $G_1,\ldots,G_\ell$ for the next recursive level are all $H'$-free for $H' = H[V(H) \setminus \set{h}]$, so the output of the recursive calls is correct. Indeed, observe that every colored copy of $H'$ in $G_i$ is completed to a colored copy of $H$ with $s_i$ acting as $h$, because $s_i$ is complete to every vertex of color $c^*$ in $G_i$. 
    Finally, note that one of the base cases must be reached as the number of vertices is reduced by one at each recursive level. 
    By union bound, the probability that a triangle survives all the way to the base case is $1 - \Order(k/n) = 1-\Order(1/n)$. Hence, the triangle is found with probability $1-\Order(1/n)$. 

    \paragraph{Analysis.} 
    Let us now prove that the running time of the algorithm, as a function of $n, \Delta_k, \ldots, \Delta_4, \Delta_3$, is:
    \begin{equation*}
        \begin{cases}
            \tOrder\left( \sum_{i=0}^{k-3} \frac{n^{2+i} \Delta_{k-i}}{\prod_{j=0}^{i-1} \Delta_{k-j}} \right) & \text{if $H$ is triangle-free,} \\ 
            \tOrder\left( \sum_{i=0}^{k-4} \frac{n^{2+i} \Delta_{k-i}}{\prod_{j=0}^{i-1} \Delta_{k-j}} \right) & \text{if $H$ contains a triangle.}
        \end{cases}
    \end{equation*}
    If the product in the denominator is empty, then it is defined to be $1$. 

    We analyze the running time of the recursive algorithm by tracking the cost at each recursive level. At the level corresponding to a pattern of size $k'$ which is not a base case, the cost is $\Order(n^2 \Delta_{k'})$ for low-degree cleanup. In addition, there is a cost of $\Order( \ell \cdot n^2)$ for constructing the instances for the next recursive level. 
    Letting $T(k')$ be the total running time for a call with a pattern of size $k'$. For $k'$, which is equal to the size of the base case (i.e. $3$ if $H$ is triangle-free, $4$ otherwise), we have that $T(3) = T(4) = \Order(n^2)$. For larger $k'$, we get:
    \[
        T(k') \leq \Order\left(n^2 \Delta_{k'} + \ell\left(n^2   + T(k'-1)\right) \right) = \tOrder\left(n^2 \Delta_{k'} + \frac{n}{\Delta_{k'}}\left(n^2   + T(k'-1)\right) \right).
    \]
    Note that $T(k'-1)$ dominates $n^2$ for every $k'$, and therefore, we write: 
    \[
    T(k') = \tOrder\left(n^2 \Delta_{k'} + \frac{n}{\Delta_{k'}} T(k'-1)\right).
    \]

    We now unroll this recurrence starting from $k$, the size of the original pattern:
    \begin{equation*}
        \begin{split}
            T(k) &= \tOrder\left(n^2 \Delta_k + \frac{n}{\Delta_k}T(k-1)\right) \\
                 &= \tOrder\left(n^2 \Delta_k + \frac{n}{\Delta_k} \left( n^2 \Delta_{k-1} + \frac{n}{\Delta_{k-1}}T(k-2)\right) \right) \\
                 &= \tOrder\left(n^2 \Delta_k + \frac{n^3 \Delta_{k-1}}{\Delta_k} + \frac{n^2}{\Delta_k \Delta_{k-1}} T(k-2)\right).
         \end{split}
    \end{equation*}
    Continuing this pattern, we get the general term:
    \begin{equation*}
        \begin{cases}
            \tOrder\left( \sum_{i=0}^{k-3} \frac{n^{2+i} \Delta_{k-i}}{\prod_{j=0}^{i-1} \Delta_{k-j}}  \right) & \text{ if $H$ is triangle-free,} \\
            \tOrder\left( \sum_{i=0}^{k-4} \frac{n^{2+i} \Delta_{k-i}}{\prod_{j=0}^{i-1} \Delta_{k-j}}  \right) & \text{ if $H$ contains a triangle.}
        \end{cases}
    \end{equation*}
    Note that we conveniently defined the degree thresholds corresponding to the base cases to be $1$, so that the final term, $T(3)$ (or $T(4)$, if $H$ contains a triangle), is accounted for by the final term of the sum, which becomes $n^2 \Delta_3$ (or $n^2 \Delta_4$, if $H$ contains a triangle). 

    \paragraph{Optimizing the parameters.}
To obtain the running time declared by the theorem, we select the degree thresholds in a manner that minimizes the maximum term in the sum. Let us provide some intuition for this optimization problem before providing the optimal assignment. We will focus on the case that $H$ is triangle-free. 
It is helpful to think of $\Delta_j$ as being equal to $n^{1-x_j}$ for every $j \in [3,k]$, for some sequence $0<x_k < x_{k-1} < \cdots < x_3 = 1$. Then, the exponents of $n$ in the terms of the above sum become equal to:
\begin{equation*}
3-x_k\,, \, 3+x_k-x_{k-1}\,,\,3+x_k + x_{k-1} - x_{k-2}\,,\,\ldots\,,\,3+(\sum_{i=4}^k x_i) - 1
\end{equation*}
This explains the exponential dependence in our running times: for $x_{k-i}$ to be larger than the sum of its predecessors, it must increase exponentially.
We can find an optimal assignment either by trial and error, or by solving the linear program: minimize $y$, subject to $3+(\sum_{i=0}^{j-1} x_{k-i}) - x_{k-j}\leq y$ for every $j \in [0,k-3]$. The optimum is achieved when we set $x_j = 1/2^{j-3}$. 

Let us now continue with the formal proof. We use the following assignment, for every $j \in [3,k]$:
\begin{equation*}
    \begin{cases}
        \Delta_j := n^{1-\frac{1}{2^{j-3}}} & \text{if $H$ is triangle-free,} \\
        \Delta_j := n^{1-\frac{1}{2^{j-4}}} & \text{if $H$ contains a triangle.}
    \end{cases}
\end{equation*}

We prove that it achieves the declared running time only for the case that $H$ is triangle-free, as the other case is nearly identical. Let us show that every term corresponding to $i \in [0,k-3]$ within the above sum, i.e., $\frac{n^{2+i} \cdot \Delta_{k-i}}{\prod_{j=0}^{i-1} \Delta_{k-j}}$, is bounded by $\tOrder(n^{3-\frac{1}{2^{k-3}}})$. For the denominator, we have: 
\[
    \prod_{j=0}^{i-1} \Delta_{k-j} = \begin{cases}
                                        1 & i=0 \\
                                        \prod_{j=0}^{i-1} n^{1-\frac{1}{2^{k-j-3}}} & i > 0
                                        \end{cases}.
\]
So for $i=0$ we have that the term corresponding to $i$ is indeed $\frac{n^{2+0} \cdot n^{1-\frac{1}{2^{k-3}}}}{1} = n^{3-\frac{1}{2^{k-3}}}$. 
For $i > 0$, we have that the exponent of $n$ in the denominator is 
\[
    \sum_{j=0}^{i-1} \left(1 - \frac{1}{2^{k-j-3}}\right) = i - \left(\frac{1}{2^{k-i-3}} - \frac{1}{2^{k-3}}\right).
\]

Thus, the term corresponding to $i$ is bounded by 
\begin{equation*}
    \begin{split}
        &\tOrder\left(\frac{n^{2+i} \Delta_{k-i}}{\prod_{j=0}^{i-1} \Delta_{k-j}}\right) = 
        \tOrder\left(\frac{n^{2+i} n^{1-\frac{1}{2^{k-i-3}}}}{n^{i-(\frac{1}{2^{k-i-3}} - \frac{1}{2^{k-3}})} }\right)= 
        \tOrder\left(n^{3-\frac{1}{2^{k-3}}}\right).
    \end{split}
\end{equation*}
Since there are $k$ terms in total and $k$ is constant, we get the declared running time. 
\end{proof}

We now proceed with adapting this algorithm to list triangles. 
\begin{proof}[Proof of \cref{thm:listing-nice}]
    We modify the algorithm from \cref{claim:lowdeg} to list all the triangles that are not incident to the colored $\Delta$-core, before returning the $\Delta$-core of the graph. This modification is straightforward. Then, the algorithm for $H$-free graphs remains the same. 
    To see that every triangle in the graph is listed, recall that in the sampling phase, a sample of $\ell = \Order( \frac{n}{\Delta_k} \log n )$ misses the neighborhood of a single triangle with probability at most $1/n^{\Theta(1)}$. By taking a sufficiently large constant factor in the definition of $\ell$, we get that the failure probability is at most $1/n^4$. Thus, by union bound, the probability that there exists a triangle whose neighborhood is missed by the sample is at most $ \binom{n}{3} / n^4 = \Order(1/n)$. Hence, we get that every triangle in the graph survives in a subgraph for the recursive level with probability $1-\Order(1/n)$. The base cases are also covered by the modification to the algorithm from \cref{claim:lowdeg}. 
\end{proof}

%% file: sections/sensitive.tex
In this section, we present $H$-sensitive algorithms, beginning with an adaptation of the algorithm for nicely colored patterns, followed by its adaptation to listing, and concluding with an efficient $C_5$-sensitive algorithm. 

\subsection{An $H$-sensitive algorithm for nicely colored patterns} 
    \label{sub:H-sensitive}
    Recall \cref{thm:H-sensitive}. 
    \Hsensitive*

    The algorithm is an adaptation of the one for $H$-free graphs from the previous section. In particular, the algorithm is also recursive, and it solves the colored variant of the problem. A first difference between the two algorithms is in their base cases.

    \begin{claim}[Base cases]
    \label{claim:sensitive-base-case}
    Let $H$ be a nicely-colored pattern, and suppose that $|V(H)| = 1$, or $|V(H)| = 4$ and $H$ contains a triangle. Then, there is an algorithm that, given a $3$-colored graph $G$ and an upper bound $1 \leq t$ on the number of colored copies of $H$ in $G$, solves the Triangle Detection problem in $\Order(n^2t)$ time. If $t =0$, then the running times are as follows. If $|V(H)| = 1$, the algorithm runs in $\Order(1)$ time, and if $|V(H)|=4$ and $H$ contains a triangle, the algorithm runs in $\Order(n^2)$ time. 
    Moreover, these running times remain the same even if the parameter $t$ provided to the algorithm is not a correct upper bound; however, the algorithm will exhibit one-sided error: it may fail to find a triangle if one exists, but it will never report a triangle otherwise.
    \end{claim}

    \begin{proof}
    If $|V(H)| = 1$, then one of the color classes in $G$ has at most $t$ vertices. We invoke the algorithm from \cref{claim:lowdeg} with $\Delta = t$ to solve the problem in $\Order(n^2 t)$ time. 
    If $|V(H)| = 4$, then $H$ must be a triangle with an attached edge (the other cases are either subgraphs of this pattern, or they are not nicely colored). Then, a triangle vertex in $G$ has at most $t$ neighbors in some color class. 
    We again invoke the algorithm from \cref{claim:lowdeg} with $\Delta = t$ to find such a triangle in $\Order(n^2 t)$ time. 

    Note that if $t$ is an incorrect upper bound, then the algorithm from \cref{claim:lowdeg} may fail to find a triangle if such a triangle is not incident to the $t$-core of the graph. 
    \end{proof}
    The following claim will also be useful during the analysis. 
    \begin{claim}
        \label{claim:H-heavy}
        Let $G$ be a graph with at most $t$ copies of a $k$-vertex pattern $H$. For every $\Delta \geq 1$, there are at most $\Delta/2$ vertices that participate in more than $2kt/\Delta$ copies of $H$. 
    \end{claim}
    \begin{proof}
        Each $H$-copy in $G$ is incident to $k$ vertices. If more than $\Delta/2$ vertices in $G$ were participating in more than $\frac{2kt}{\Delta}$ copies of $H$, there would be more than $\frac{\frac{\Delta}{2} \cdot \frac{2kt}{\Delta}}{k} = t$ distinct copies of $H$, contradicting the upper bound of $t$. 
    \end{proof}

We define the input to the recursive algorithm to be a five-tuple $(G, C_G, H, C_H, t)$, where the first four elements are defined the same as before, but now it also contains an upper bound $t$ on the number of copies of $H$ in $G$. 

    \begin{proof}[Proof of \cref{thm:H-sensitive}]
        The recursive algorithm is nearly the same as before, with the differences being only in the first and last steps. Namely, in each recursive level, the following procedure is applied:
    \begin{itemize}
        \item \textbf{Base-case:} 
        If $H$ contains a triangle and $k=4$, or if $k=1$, invoke the algorithm from \cref{claim:sensitive-base-case} on $G$ and $t$ and return its output. 

        \item \textbf{Low-degree cleanup:} 
        This step remains unchanged. Either a triangle is found, or the algorithm proceeds with a colored $\Delta_k$-core of $G$, denoted $G'$.

        \item \textbf{Random embedding of $h$ in $G$:} 
        This step remains unchanged. We obtain $\ell := \Theta( \frac{n}{\Delta_k}  \cdot \log n)$ subgraphs of $G$, denoted $G_1,\ldots,G_\ell$. 
        
        \item \textbf{Recursive calls:}  
        For every $i \in [\ell]$, make a recursive call on 
        $(G_i, C_{G_i}, H', C_{H'}, t')$, where $t'= 2kt/\Delta_k$ and $H'$ is the remaining pattern after removing $h$. 
    \end{itemize}
    \paragraph{Correctness.} Let us prove that if $G$ contains a triangle, the algorithm finds one with probability $1 - \Order(1/n)$. The algorithm never reports a triangle otherwise because both its subroutines, \cref{claim:sensitive-base-case} and \cref{claim:lowdeg} do not do so (even if a base case is reached with an incorrect value of $t$). 

    The only difference from the correctness proof in the $H$-free setting is that it needs to be shown that with probability $1-\Order(1/n)$, there exists a subgraph $G_i$ that contains a triangle, and also contains at most $t'$ colored copies of $H'$. This claim implies that if a triangle survives through all recursive levels, it must be found in the base case. 
    Suppose that a triangle does not contain vertices outside the $\Delta_k$-core, so it is not found in the low-degree cleanup. By \cref{claim:H-heavy}, there are at most $\Delta_k/2$ vertices in the entire graph that are ``$H$-heavy'' (i.e., participate in more than $2kt/\Delta_k$ copies of $H$). A triangle that survives a cleanup step must have a vertex with at least $\Delta_k$ neighbors. At most $\Delta_k/2$ of these neighbors can be $H$-heavy. Therefore, the triangle has at least $\Delta_k - \Delta_k/2 = \Delta_k/2$ ``$H$-light'' neighbors. Hence, the probability that $h$ is embedded into one of the vertices that is incident to less than $2kt/\Delta_k$ copies of $H$, is at least 
    \[
        1-\left(1-\frac{\Delta/2}{n}\right)^{\ell} \geq 1-\exp\left(- \frac{\Delta}{2n}\cdot \ell\right) = 1-n^{- \Theta(1)}.
    \]
    By choosing an appropriate constant factor in $\ell$, we get that this happens with probability $1-\Order(1/n)$. Suppose now that $s_i$ is an $H$-light neighbor of the triangle in the sample, i.e., $s_i$ participates in at most $t' = 2kt/\Delta_k$ colored copies of $H$. So $G_i$ contains the triangle, and the number of colored copies of $H'$ in $G_i$ is at most $t'$. That is because every colored copy of $H'$ in $G_i$ is completed, with $s_i$, to a colored copy of $H$ in $G$.

    \paragraph{Analysis.}
    Let us focus on the case where $H$ is triangle-free, and we will consider the other case later. 
    Our goal is to balance the cost of the recursive calls with the cost of the base case, which depends on $t$. We define $\Delta_1 := \frac{t}{\prod_{j=2}^k \Delta_j}$ to represent the cost of this final base case. Let us prove that the running time of the algorithm is:
    \begin{equation*}
            \tOrder\left( \frac{n^{k+1}}{\prod_{j=2}^k \Delta_{j}}  + \sum_{i=0}^{k-1} \frac{n^{2+i} \Delta_{k-i}}{\prod_{j=0}^{i-1} \Delta_{k-j}} \right).
    \end{equation*}
    Let $T(k',t')$ be the total running time for a call on a pattern of size $k'$ and an upper bound $t'\geq 0$ on the number of colored copies of that pattern. It follows from \cref{claim:sensitive-base-case} that $T(1,t') = \Order(n^2 t')$. For a larger $k'$, we get:
    \[
        T(k',t') = 
        \tOrder\left(n^2 \Delta_{k'} + \frac{n}{\Delta_{k'}}\left(n^2   + T\left(k'-1,\frac{2t'k'}{\Delta_{k'}}\right)\right) \right).
    \]
    Note that unlike in the analysis of the $H$-free setting, here the $n^2$ term is not necessarily dominated by $T(k'-1,\frac{2t'k'}{\Delta_{k'}})$, because for $k' = 2$ it could be that $\frac{t2k'}{\Delta_{k'}} < 1$ (which is treated as $0$), and thus $T(1,\frac{t2k'}{\Delta_{k'}}) = \Order(1)$. Hence, when we unroll this recurrence starting from $k$ and $t$, we get:
\begin{equation*}
    \begin{split}
        T(k,t) &= \tOrder\left(n^2 \Delta_k + \frac{n}{\Delta_k}\left(n^2 + T\left(k-1,\frac{t2k}{\Delta_k}\right) \right) \right) \\
             &= \tOrder\left(n^2 \Delta_k + \frac{n}{\Delta_k} \left( n^2 + n^2\Delta_{k-1} + \frac{n}{\Delta_{k-1}}\left(n^2 + T\left(k-2,\frac{t2k\cdot2(k-1)}{\Delta_k\Delta_{k-1}}\right)\right)\right) \right)  \\
             &=\tOrder\left( n^2 \Delta_k + \frac{n^3}{\Delta_k} + \frac{ n^3 \Delta_{k-1}}{\Delta_k} + \frac{n^4}{\Delta_k \Delta_{k-1}} + \frac{n^4 \Delta_{k-2}}{\Delta_k \Delta_{k-1}}  + \cdots + \frac{n^{k+1}}{\prod_{i=2}^{k} \Delta_i} + \frac{n^{k-1} \cdot T\left(1,\frac{t \prod_{i=0}^{k-2} (2 (k-i)) }{\prod_{i=2}^k\Delta_i}\right)}{\prod_{i=2}^{k} \Delta_i} \right) \\
             &=\tOrder\left( n^2 \Delta_k + \frac{n^3}{\Delta_k} + \frac{ n^3 \Delta_{k-1}}{\Delta_k} + \frac{n^4}{\Delta_k \Delta_{k-1}} + \frac{n^4 \Delta_{k-2}}{\Delta_k \Delta_{k-1}}  + \cdots + \frac{n^{k+1}}{\prod_{i=2}^{k} \Delta_i} + \frac{n^{k+1} \Delta_1 }{\prod_{i=2}^{k} \Delta_i} \right) \\
             &= \tOrder\left( \frac{n^{k+1}}{\prod_{i=2}^k \Delta_{i}}  + \sum_{i=0}^{k-1} \frac{n^{2+i} \Delta_{k-i}}{\prod_{j=0}^{i-1} \Delta_{k-j}} \right)
     \end{split}
\end{equation*}
    Here, we conveniently used the definition of $\Delta_1 = \frac{t}{\prod_{i=2}^k \Delta_i}$ because the running time for the base case, starting with a pattern of size $k$ and an upper bound $t$, is:
    \[
T\left(1,\frac{t \prod_{i=0}^{k-2} (2 (k-i)) }{\prod_{i=2}^k\Delta_i}\right) = 
         \Order\left( n^2 \Delta_1 \right).
    \]

    Now consider the case where $H$ contains a triangle. Let us prove that the running time of the algorithm is:
    \[
        \tOrder \left( \frac{n^{k-2}}{\prod_{i=5}^k \Delta_{i}} + \sum_{i=0}^{k-4} \frac{n^{2+i} \Delta_{k-i}}{\prod_{j=0}^{i-1} \Delta_{k-j}} \right), 
    \]
    where we define the degree threshold corresponding to the base case as $\Delta_4 = \frac{t}{\prod_{i=5}^k \Delta_i}$. 
    Here, because the running time for the base case (\cref{claim:sensitive-base-case}) is $T(4,t') = \Order(n^2(1+t'))$ for every $t' \geq 0$, the $n^2$ terms are dominated by $T(k'-1,\frac{2t'k'}{\Delta_{k'}})$. Hence, similarly to the analysis in the $H$-free setting, we get: 
\begin{equation*}
    \begin{split}
             &\tOrder\left( n^2 \Delta_k +  \frac{ n^3 \Delta_{k-1}}{\Delta_k}  + \frac{n^4 \Delta_{k-2}}{\Delta_k \Delta_{k-1}}  + \cdots +  \frac{n^{k-4} \cdot T\left(4,\frac{t \prod_{i=0}^{k-5} (2 (k-i)) }{\prod_{i=5}^k\Delta_i}\right)}{\prod_{i=5}^{k} \Delta_i} \right) \\
             &=\tOrder\left( n^2 \Delta_k +  \frac{ n^3 \Delta_{k-1}}{\Delta_k}  + \frac{n^4 \Delta_{k-2}}{\Delta_k \Delta_{k-1}}  + \cdots +  \frac{n^{k-4} \cdot (n^2 + n^2\Delta_4) }{\prod_{i=5}^{k} \Delta_i} \right) \\ 
             &= \tOrder \left( \frac{n^{k-2}}{\prod_{i=5}^k \Delta_{i}} + \sum_{i=0}^{k-4} \frac{n^{2+i} \Delta_{k-i}}{\prod_{j=0}^{i-1} \Delta_{k-j}} \right).
    \end{split}
\end{equation*}

    \paragraph{Optimizing the parameters.}
    Let us start with the case that $H$ is triangle-free. Recall the running time:
    \begin{equation*}
            \tOrder\left( \frac{n^{k+1}}{\prod_{j=2}^k \Delta_{j}}  + \sum_{i=0}^{k-1} \frac{n^{2+i} \Delta_{k-i}}{\prod_{j=0}^{i-1} \Delta_{k-j}} \right).
    \end{equation*}
    As we did in the previous section, it is helpful to think of $\Delta_j$ as $n^{1-x_j}$ for every $j \in [2,k]$. 
    The exponent in the leftmost term is  
    \[
        (k+1)-\left(\sum_{i=2}^k(1-x_i)\right)= 2+\sum_{i=2}^k x_i. 
    \]
    For the rightmost term, i.e., the one corresponding to $\Delta_1$, let us write $t = n^{z}$ for some $0 \leq z \leq k$. \footnote{We may assume that $t \geq 1$, and hence $z \geq 0$, as otherwise one can simply invoke the algorithm for $H$-free graphs from the previous section.}
    So the exponent of $\Delta_1 = \frac{t}{\prod_{j=2}^k \Delta_j}$ is $z - \sum_{i=2}^k (1-x_i) = z -(k-1) + \sum_{i=2}^k x_i$. Together, the exponents of all the terms in the running time are divided into three groups: 
    \begin{equation}
        \label{eq:exp1}
        2+\sum_{i=2}^k x_i,
    \end{equation}
    \begin{equation}
        \label{eq:exp2}
        3-x_k\,, \, 3+x_k-x_{k-1}\,,\,\ldots\,,\,3+\left(\sum_{i=3}^k x_i\right) -x_2, 
    \end{equation}
    \begin{equation}
        \label{eq:exp3}
        3 +z-k+ 2\left(\sum_{i=2}^k x_i\right).
    \end{equation}

    Observe that for sufficiently small $z$, \cref{eq:exp3} is dominated by \cref{eq:exp1} and \cref{eq:exp2}. 
    In this case, \cref{eq:exp1} and \cref{eq:exp2} are optimized to $3-\frac{1}{2^{k-1}}$ when we set $x_j = \frac{1}{2^{j-1}}$ for every $j \in [2,k]$. For larger $z$, \cref{eq:exp3} dominates \cref{eq:exp1}. In this case, \cref{eq:exp2}  and \cref{eq:exp3} optimize to $3-\frac{k-z}{2^k-1}$ when we set $x_j = \frac{C(k-z)}{2^{j-1}}$, where $C = \frac{2^{k-1}}{2^k-1}$. The cut-off occurs at $z = k-2+\frac{1}{2^{k-1}}$. Thus, let us state the final assignments. For every $j \in [2,k]$, we set: 
    \begin{equation*}
        \begin{cases}
            \Delta_j = n^{1-\frac{1}{2^{j-1}}} & t \leq n^{k-2+\frac{1}{2^{k-1}}} \\
            \Delta_j = n^{1-\frac{kC}{2^{j-1}}} \cdot t^{\frac{C}{2^{j-1}}} & t > n^{k-2+\frac{1}{2^{k-1}}}.
        \end{cases}
    \end{equation*}

            Let us now prove that these assignments indeed give the stated upper bound. Throughout the proof, we will use \cref{eq:exponents} without stating it each time. \\
            \textbf{Case 1:} $t \leq n^{k-2+\frac{1}{2^{k-1}}}$. \cref{eq:exp1} is:
            \begin{equation*}
                \begin{split}
                    2+\sum_{i=2}^k x_i = 2 + \sum_{i=2}^k \frac{1}{2^{i-1}} = 2 + \sum_{i=1}^{k-1} \frac{1}{2^i} = 2 + (1 - \frac{1}{2^{k-1}}) = 3 - \frac{1}{2^{k-1}}.
                \end{split}
            \end{equation*}
            The terms in \cref{eq:exp2}, for every $j \in [1,k-2]$, are:
            \begin{equation*}
                \begin{split}
                3 + \left(\sum_{i=0}^{j-1} x_{k-i}\right) - x_{k-j} &= 3 + \left(\sum_{i=0}^{j-1} \frac{1}{2^{k-i-1}}\right) - \frac{1}{2^{k-j-1}}\\ 
                &= 3 + \left(\sum_{i=k-j}^{k-1} \frac{1}{2^{i}}\right) - \frac{1}{2^{k-j-1}} \\
                &=3 + (\frac{1}{2^{k-j-1}} - \frac{1}{2^{k-1}}) - \frac{1}{2^{k-j-1}} \\
                &= 3 - \frac{1}{2^{k-1}}
                \end{split}
            \end{equation*}
            Moreover, for \cref{eq:exp3}, we have:
            \begin{equation*}
                \begin{split}
                3 +z-k+ 2\left(\sum_{i=2}^k x_i\right) &\leq 
                3 + \left(k-2 + \frac{1}{2^{k-1}}\right) - k + 2\left(\sum_{i=2}^k \frac{1}{2^{i-1}}\right) \\ 
                &= 1 + \frac{1}{2^{k-1}} + 2\left(1- \frac{1}{2^{k-1}}\right) = 3- \frac{1}{2^{k-1}}.
                \end{split}
            \end{equation*}

            \textbf{Case 2:} $t > n^{k-2+\frac{1}{2^{k-1}}}$. Again, we compute each term. 
\cref{eq:exp1} is:
            \begin{equation*}
                \begin{split}
                    2+\sum_{i=2}^k x_i &= 2 + \sum_{i=2}^k \frac{C(k-z)}{2^{i-1}} = 2 + C(k-z)(1 - \frac{1}{2^{k-1}}) \\
                    &=2 +  (k-z)\frac{2^{k-1}-1}{2^k-1} \\
                    &= 2 + (k-z) \frac{2^{k-1}}{2^k-1} - \frac{k-z}{2^k-1} \\
                    &\leq 2 + \left(k-\left(k-2+\frac{1}{2^{k-1}}\right)\right)\frac{2^{k-1}}{2^k-1} - \frac{k-z}{2^k-1} \\
                    &= 3 - \frac{k-z}{2^k-1}.
                \end{split}
            \end{equation*}
            The terms in \cref{eq:exp2}, for every $j \in [1,k-2]$, are:
            \begin{equation*}
                \begin{split}
                3 + \left(\sum_{i=0}^{j-1} x_{k-i}\right) - x_{k-j} &= 3 + \left(\sum_{i=0}^{j-1} \frac{C(k-z)}{2^{k-i-1}}\right) - \frac{C(k-z)}{2^{k-j-1}}
                = 3 - \frac{C(k-z)}{2^{k-1}} 
                = 3-\frac{k-z}{2^k-1}
                \end{split}
            \end{equation*}
            Moreover, for \cref{eq:exp3}, we have:
            \begin{equation*}
                \begin{split}
                3 +z-k+ 2\left(\sum_{i=2}^k x_i\right) &= 
                3 + z - k + 2\left(\sum_{i=2}^k \frac{C(k-z)}{2^{i-1}}\right) \\ 
                &= 3+ z-k + 2C(k-z) \left(1-\frac{1}{2^{k-1}}\right)  \\
                &= 3 + z - k + (k-z) \left(1- \frac{1}{2^k-1} \right) \\
                &= 3- \frac{k-z}{2^k-1}.
                \end{split}
            \end{equation*}

            We have shown that if $t \leq n^{k-2+\frac{1}{2^{k-1}}}$, the exponents of all the terms are at most $3-\frac{1}{2^{k-1}}$, and if $t$ is larger, then the exponents are at most $3-\frac{k-z}{2^{k}-1}$. This concludes the proof of the first part of the theorem: if $H$ is triangle-free, then there is a Triangle Detection algorithm for $H$-free graphs, running in time: 
            \[
                \tOrder\left(n^{3-\frac{1}{2^{k-1}}} + n^{3-\frac{k-z}{2^k-1}}\right) =
\tOrder\left( n^{3-\frac{1}{2^{k-1}}} + n^{3 - \frac{k}{2^k - 1}} \cdot t^{\frac{1}{2^k - 1}}\right) 
            \]

        Let us now continue with the case that $H$ contains a triangle. Recall that the running time of the algorithm is:
    \[
        \tOrder \left( \frac{n^{k-2}}{\prod_{i=5}^k \Delta_{i}} + \sum_{i=0}^{k-4} \frac{n^{2+i} \Delta_{k-i}}{\prod_{j=0}^{i-1} \Delta_{k-j}} \right). 
    \]
        In this case, the exponents are the following, writing $t = n^z$ for some $0\leq z \leq k-3$:
        \begin{equation*}
        \begin{split}
             &2 + \sum_{i=5}^k x_i, \\ 
             & \forall j \in [0,k-5]:\; 3 + \left(\sum_{i=0}^{j-1} x_{k-i}\right) - x_{k-j},\\
             &6 + z-k + 2\left(\sum_{i=5}^k x_i\right)
        \end{split}
        \end{equation*}
        Solving the optimization problem similarly to before yields the following optimal assignments.
        If $z \leq k-5 + \frac{1}{2^{k-4}}$, we set $x_i = \frac{1}{2^{i-4}}$ for every $i \in [5,k]$. 
        Otherwise, we set $x_j = \frac{C(k-3-z)}{2^{j-4}}$, where $C = \frac{2^{k-4}}{2^{k-3}-1}$.           

        \textbf{Case 1:} $t < n^{k-5+\frac{1}{2^{k-4}}}$. 
        We show, using \cref{eq:exponents}, that every exponent above is bounded by $3-\frac{1}{2^{k-4}}$. For the first one, we get: 
        \[
            2 + \sum_{i=5}^k x_i  = 2 + \sum_{i=5}^k \frac{1}{2^{i-4}} = 2 + \left(1-\frac{1}{2^{k-4}}\right) = 3- \frac{1}{2^{k-4}}.
        \]
        For the middle exponents, we get
        \begin{equation*}
        \begin{split}
            &3 + \left(\sum_{i=0}^{j-1} x_{k-i}\right) - x_{k-j} = 3 + \left(\sum_{i=0}^{j-1} \frac{1}{2^{k-i-4}}\right) - \frac{1}{2^{k-j-4}} \\
            &= 3 + \left(\frac{1}{2^{k-j-4}} - \frac{1}{2^{k-4}}\right) - \frac{1}{2^{k-j-4}} = 3 - \frac{1}{2^{k-4}}.
        \end{split}
        \end{equation*}
        For the last exponent, we get:
        \begin{equation*}
        \begin{split}
            &6 + z - k + 2\left(\sum_{i=5}^k x_i\right) = 
            6 + z - k + 2\left(1-\frac{1}{2^{k-4}} \right) \\
            &\leq 6 + \left(k-5+\frac{1}{2^{k-4}}\right) -k + 2\left(1-\frac{1}{2^{k-4}} \right) \\
            &= 3 - \frac{1}{2^{k-4}}
        \end{split}
        \end{equation*}

        \textbf{Case 2:} $t \geq n^{k-5+\frac{1}{2^{k-4}}}$. Let us show that every exponent is at most $3- \frac{(k-3-z)}{2^{k-3}-1}$. For the first term, we have:
        \begin{equation*}
        \begin{split}
            &2 + \sum_{i=5}^k x_i = 2 + \sum_{i=5}^k \frac{C(k-3-z)}{2^{i-4}} = 2 + C(k-3-z)\left(1-\frac{1}{2^{k-4}}\right) \\
            &= 2 + (k-3-z)\frac{2^{k-4}-1}{2^{k-3}-1} 
            = 2 + (k-3-z) \frac{2^{k-4}}{2^{k-3}-1} - \frac{k-3-z}{2^{k-3}-1} \\
            &\leq 2+\left(k-3-\left(k-5+\frac{1}{2^{k-4}}\right)\right)\frac{2^{k-4}}{2^{k-3}-1} - \frac{k-3-z}{2^{k-3}-1}  \\
            &= 3- \frac{k-3-z}{2^{k-3}-1}.
        \end{split}
        \end{equation*}

        For the middle terms, we have:
        \begin{equation*}
            \begin{split}
            &3 + (\sum_{i=0}^{j-1} x_{k-i}) - x_{k-j} = 3 + \left(\sum_{i=0}^{j-1} \frac{C(k-3-z)}{2^{k-i-4}}\right)  - \frac{C(k-3-z)}{2^{k-j-4}} \\
            &=3 - C(k-3-z) \frac{1}{2^{k-4}} \\
            &= 3- \frac{k-3-z}{2^{k-3}-1}.
            \end{split}
        \end{equation*}
        Moreover, the last exponent is:
        \begin{equation*}
        \begin{split}
            &6 + z - k + 2(\sum_{i=5}^k x_i) = 6 + z - k + 2\left(\sum_{i=5}^k \frac{C(k-3-z)}{2^{i-4}}\right) \\
            &=6+ z - k + (k-3-z)C\cdot 2\left(1- \frac{1}{2^{k-4}}\right) \\
            &= 6+ z - k + (k-3-z) \frac{2^{k-3}-2}{2^{k-3}-1} \\
            &= 6+z-k +(k-3-z) - \frac{k-3-z}{2^{k-3}-1} = 3 -  \frac{k-3-z}{2^{k-3}-1}. 
        \end{split}
        \end{equation*}
        We have shown that the running time of the algorithm, when $H$ contains a triangle, is 
        \[
            \tOrder\left(n^{3-\frac{1}{2^{k-4}}} + n^{3-\frac{k-3-z}{2^{k-3}-1}}\right) = \tOrder\left(n^{3-\frac{1}{2^{k-4}}} + n^{3-\frac{k-3}{2^{k-3}-1}} \cdot t^{\frac{1}{2^{k-3}-1}}\right),
        \]
        which concludes the proof.
    \end{proof}

    The algorithm is also adapted to list the triangles. 
    \begin{proof}[Proof of \cref{thm:listing-sensitive}]
        Adapting the above algorithm also to list all triangles is done similarly to what was done for listing in $H$-free graphs. The base cases are clearly covered because adapting the algorithm from $\cref{claim:lowdeg}$ to list triangles is easy. The only difference is that we need to argue that for every triangle in the graph, an $H$-light neighbor of that triangle is hit by the sample with high probability. Since the number of $H$-light neighbors of that triangle is $\Omega(\Delta)$, where $\Delta$ is the degree threshold used in a recursive level, the claim holds. 
    \end{proof}

\subsection{A $C_5$-sensitive algorithm}
    In this subsection, we prove \cref{thm:C5-sensitive}.
    \Cfivesensitive*

    Let us start with some definitions. Let $1 \leq \Delta \leq n$ be a parameter. 
    For every vertex $v \in V(G)$, let $L_1(v) := \set{u \in N(v) | \deg(u) \leq \Delta}$, and $L_2(v) := \set{ u \in N(L_1(v)) \setminus \set{v} | \deg(u) \leq \Delta} $. For every $u \in L_2(v)$, denote by $p_v(u)$ the number of neighbors that $u$ has in $L_1(v)$. Denote by $t_v$ the number of $5$-cycles containing $v$.

    \begin{claim}
    \label{claim:pxpy}
        Let $v \in V(G)$, and suppose that there are no triangles in $ L_1(v) \cup L_2(v)$ that consist of one vertex in $L_1(v)$ and two vertices in $L_2(v)$. Then:
        \[
            \sum_{uw \in E(L_2(v),L_2(v))} p_v(u)p_v(w) \leq t_v.
        \]
    \end{claim}
    \begin{proof}
        The sum $\sum_{uw \in E(L_2(v),L_2(v))} p_v(u)p_v(w)$ counts the number of $4$-vertex walks of the form $s - u - w - s'$ such that $s, s' \in L_1(v)$ and $uw$ is an edge in $L_2(v)$. If $s = s'$, the walk is a triangle with a vertex in $L_1(v)$ and two vertices in $L_2(v)$. By assumption, triangles of this type do not exist. If $s \neq s'$, the walk completes a $5$-cycle $v - s - u - w - s' - v$.  Hence, every walk counted in the sum corresponds to a unique $5$-cycle containing $v$. Thus, the sum is at most the number of $5$-cycles containing $v$. 
    \end{proof}

    The following lemma will be crucial for the analysis of the algorithm. The reader can skip its proof and return to it later. 
    
    \begin{lemma}
    \label{lem:min}
        Let $v \in V(G)$, and suppose that there are no triangles in $ L_1(v) \cup L_2(v)$ that consist of one vertex in $L_1(v)$ and two vertices in $L_2(v)$. Then:
        \[
            \sum_{uw \in E(L_2(v),L_2(v))} \min( p_v(u), p_v(w) ) = \Order( n^{2/3} t_v^{1/3} \Delta ).
        \]
    \end{lemma}
    \begin{proof}
        Let $\tau$ be a parameter to be set later. Define $light(v) := \set{ uw \in E(L_2(v),L_2(v)) \mid p_v(u)p_v(w) \leq \tau }$ and $heavy(v) := \set{ uw \in E(L_2(v),L_2(v)) \mid p_v(u)p_v(w) > \tau }$.
        For heavy edges, i.e., $uw \in heavy(v)$, we get from \cref{claim:pxpy} that $\sum_{uw \in heavy(v)} p_v(u)p_v(w) \leq t_v$, and therefore the number of heavy edges is bounded by $|heavy(v)| \leq t_v/\tau$. For those edges we use the bound $\min( p_v(u), p_v(w) ) \leq \min(\deg(u), \deg(w)) \leq \Delta$. 
        For light edges, i.e., $uw \in light(v)$, we have $\min( p_v(u),p_v(w) ) \leq \sqrt{\tau}$. We bound their number by $|E(L_2(v),L_2(v))| \leq n \Delta$. Thus, we get:
        \begin{equation*}
            \begin{split}
                &\sum_{uw \in E(L_2(v),L_2(v))} \min(p_v(u),p_v(w)) = 
                \sum_{uw \in light(v)} \min( p_v(u), p_v(w) ) + \sum_{uw \in heavy(v)} \min(p_v(u),p_v(w)) \\ 
                &\leq |E(L_2(v),L_2(v))|\sqrt{\tau} + \frac{t_v}{\tau} \Delta 
                = \Order\left(n \Delta \sqrt{\tau} + \frac{t_v}{\tau} \Delta \right).
            \end{split}
        \end{equation*}
        This sum is minimized to $\Order(n^{2/3} t_v^{1/3} \Delta )$ by setting $\tau = (t_v /n)^{2/3}$.
    \end{proof}

\begin{proof}[Proof of \cref{thm:C5-sensitive}]
    Suppose that we know the maximum degree $\Delta$ of a vertex in a triangle. We will later replace this assumption with exponential search. The algorithm is as follows:
    \begin{enumerate}[label=(\roman*)]
        \item 
            Sample uniformly at random $\ell = \Theta( \frac{n}{\Delta} \log n)$ vertices from $G$. For each sampled vertex $v$, process its adjacency list and find the vertices in $L_1(v)$. For every $u \in L_1(v)$, process its adjacency list and find its neighbors in $L_2(v)$. If a triangle containing $v$ is found in this step (i.e., an edge within $N(v)$), report the triangle and halt. Otherwise, find all the edges with both endpoints in $L_2(v)$ by processing the adjacency lists of the vertices in $L_2(v)$. 
        \item 
            For every sampled vertex $v$, and for every edge $uw \in E(L_2(v),L_2(v))$, check if $uw$ participates in a triangle with some vertex in $L_1(v)$. This step is implemented as follows. For the endpoint with fewer neighbors in $L_1(v)$ (i.e., the endpoint that minimizes $p_v(\cdot)$), check if any of its neighbors in $L_1(v)$ are also adjacent to the other endpoint. If so, report a triangle and halt. 
    \end{enumerate}
    If no triangle is found throughout the execution, report that $G$ is triangle-free.
    \paragraph{Correctness.} Let us now prove that if $G$ contains a triangle, the algorithm finds it with probability $1-\Order(1/n)$. It is clear that the algorithm never reports a triangle otherwise. Suppose that $G$ contains a triangle such that the maximum degree of its vertices is $\Delta$. Hence, the probability that a sample of $\ell$ vertices hits the neighborhood of the maximum degree vertex in the triangle is at least:
    \[
        1-\left(1-\frac{\Delta}{n}\right)^{\ell} \geq 1- \exp \left(- \frac{\Delta}{n} \ell\right) = 1- \exp \left( \Theta\left(\frac{\Delta}{n} \cdot \frac{n}{\Delta} \log n \right)\right) = 1-n^{-\Theta(1)}.
    \]
    By choosing an appropriate constant factor in $\ell$, this event occurs with probability $ 1-\Order(1/n)$. Now suppose that $v$ is a sampled vertex that is also a neighbor of a triangle vertex. Suppose that $v$ is not in a triangle, as that case is handled in step (i). Since the maximum degree among the triangle vertices is $\Delta$, and $v$ is not in a triangle, then $L_1(v)$ contains one vertex of the triangle and $L_2(v)$ contains two. Therefore, the triangle must be found in step (ii).

    \paragraph{Analysis} Let us prove that the expected running time is $\tOrder(n^2 + n^{4/3} t^{1/3})$. We later show how to achieve worst-case running time at the cost of an extra $\log n$-factor. 
    The cost of sampling $\ell$ vertices and computing $L_1(v)$ and $L_2(v)$ for every sampled vertex $v$ is: 
    \[
    \tOrder\left(\ell n\Delta) \right) = \tOrder(n^2).
    \]
    For step (ii), searching for a triangle in $L_1(v) \cup L_2(v)$, the cost is bounded by \cref{lem:min}:
    \[
        \Order\left(\sum_{uw \in E(L_2(v),L_2(v))} \min(p_v(u),p_v(w))\right)=\Order(n^{2/3} t_v^{1/3} \Delta ).
    \]
    Hence, the total expected time for step (ii), where the probability is taken over the choice of $\ell$ i.i.d. vertices $v_1,\ldots,v_\ell$, is:
    \[
        \Ex\left[\sum_{j=1}^\ell n^{2/3} t_{v_j}^{1/3} \Delta\right] = \ell n^{2/3} \Delta \Ex\left[  t_{v_1}^{1/3} \right]. 
    \]
    By Jensen's inequality, we have that:
    \[
        \Ex\left[ t_{v_1}^{1/3} \right] \leq \Ex\left[ t_{v_1} \right]^{1/3} = \Order\left( \left(\frac{t}{n} \right)^{1/3} \right).
    \]
    Hence, the total expected time is:
    \[
        \tOrder\left( n^2 + \ell n^{2/3} \Delta \left( \frac{t}{n} \right)^{1/3} \right) = \tOrder(n^2 + n^{4/3} t^{1/3} ). 
    \]
    To remove the assumption on the knowledge of $\Delta$, we execute the algorithm with $\Delta_i = 2^i$ for every $i \in [0, \lceil \log n \rceil]$. Note that all the above claims also hold when using $\Delta_i$, for $\Delta \leq \Delta_i  \leq 2 \Delta$. This adds a $\log n$-factor in the running time. 

    To make the algorithm run in worst-case time (and not in expectation), we execute $\Order(\log n)$ independent copies of the algorithm, but we cap their running time at $10$ times the expectation. An execution that reaches the cap is considered failed. We return the answer of the first execution that does not fail. The probability that all executions fail is, by Markov's inequality, at most $(1/10)^{\Order( \log n )} = \Order(1/n)$, by choosing an appropriate constant. This adds another $\log n$-factor in the running time. 
    \end{proof}

%% file: sections/attached.tex
In this section, we prove \cref{thm:attached} and \cref{thm:listing-attached}. Let us recall the statement of the first theorem.
\attached*

We begin by showing an algorithm for colored patterns that are only almost nicely colored: the only obstacle to a nice coloring is a single triangle vertex. Later, we combine this algorithm with the embedding algorithm from the proof of \cref{thm:H-nice}, to prove \cref{thm:attached}. 
\begin{lemma}
    \label{lem:coloredH}
    Let $H$ be a colored pattern satisfying the following property. There is a triangle $xyz$ in $H$ such that $z$ is of degree $2$, and $H[V(H)\setminus\set{z}]$ is nicely colored. Then there is a combinatorial algorithm for Triangle Detection on colored $H$-free graphs that, given a $3$-colored graph $G$, runs in time $\tOrder(n^{3-\frac{1}{2^{k-1}}})$ and succeeds with probability $1-\Order(1/n)$. 
\end{lemma}

\begin{proof}
    Denote by $C_G$ and $C_H$ the colorings of $G$ and $H$, respectively. We begin with a sparsification procedure, removing colored copies of $H' = H[V(H)\setminus\set{z}]$ from the graph. 
    Let $\tau$ be a parameter to be set later. For every edge $uv \in E(G)$ whose colors match the colors of $xy$, namely, $C_G(u) = C_H(x)$ and $C_G(v) = C_H(y)$, apply the following procedure. 
    Sample uniformly at random $\ell = \Order(\frac{n^{k-3}}{\tau} \cdot \log n)$ multisets of $k-3$ vertices from $G$. For each multiset $S$, consider the subgraph induced by $\set{u,v} \cup S$. If it forms a colored copy of $H'$, where $uv$ is mapped to $xy$, then check if $uv$ is in a triangle with some vertex in $S$. If so, report the triangle and halt. Otherwise, delete $uv$ from $G$. 

    Let $G'$ denote the remaining subgraph of $G$. We now invoke the recursive algorithm from the proof of \cref{thm:H-sensitive} on the input $(G', C_G, H', C_{H'}, n^2 \tau)$, where $C_{H'}$ is $C_{H}$ restricted to $H'$. If it returns a triangle, we are done. Otherwise, report that $G$ is triangle-free.  

    \paragraph{Correctness.} Let us prove that if $G$ contains a triangle, then the algorithm finds one with probability $1-\Order(1/n)$. It is clear that the algorithm never reports a triangle otherwise. Suppose that $G$ contains a triangle $uvw$ and that the colors of $uv$ match the colors of $xy$. Then $uv$ does not participate as $xy$ in a colored copy of $H'$ that is disjoint from $w$, or else we would get a colored copy of $H$ in $G$. Hence, if for some set of vertices $S$ we have that $S \cup \set{uv}$ induces a colored copy of $H'$ with $uv$ acting as $xy$, then $w$ must belong to $S$. This shows that the sparsification procedure is correct, i.e., it never deletes an edge that participates in a triangle. 

    It remains to show that with probability $1-\Order(1/n)$, the number of colored copies of $H'$ in $G'$ is at most $\tau n^2$, so correctness follows from that of the $H'$-sensitive algorithm.
    Suppose that an edge $e$ participates in more than $\tau$ colored copies of $H'$ where $e$ is mapped to $xy$. Every such copy contains $(k-3)$ vertices other than the endpoints of $e$. 
    Hence, the probability that a single sample hits a colored copy of $H'$ containing $e$ as $xy$, is at least $\frac{\tau}{n^{k-3}}$. The probability that none of the $\ell$ samples hits such a copy is at most:
    \[
        \left(1 - \frac{\tau}{n^{k-3}} \right)^\ell \leq \exp\left( -\frac{\tau}{n^{k-3}} \cdot \ell \right) = \frac{1}{n^{\Theta(1)}}.
    \]
    By choosing an appropriate constant factor for $\ell$, we get from the union bound that with probability $1-\Order(1/n)$, all the edges that participate in more than $\tau$ colored copies of $H'$ as $xy$ are discovered and removed (or a triangle is found). So by the end of the sparsification procedure, the number of colored copies of $H'$ in $G$ is bounded by:
    \[
        \sum_{e \in E(G')} \text{(\#colored copies of $H'$ where $e$ is mapped to $xy$)} \leq n^2 \tau.
    \]
    Hence, the invocation of the $H'$-sensitive algorithm uses a correct upper bound with probability $1-\Order(1/n)$.  

    \paragraph{Analysis.}
    The total running time of the algorithm is given by the time to process the edges of $G$ and by the running time of the $H'$-sensitive algorithm with $t = n^2 \tau$. We get:
    \begin{equation*}
        \begin{split}
            &\tOrder\left(n^2 \ell + n^{3-\frac{1}{2^{k-1}}} + n^{3 - \frac{k}{2^k - 1}} \cdot t^{\frac{1}{2^k - 1}}\right) = 
                         \tOrder\left( \frac{n^{k-1}}{\tau} + n^{3-\frac{1}{2^{k-1}}}+ n^{3 - \frac{k}{2^k - 1}} \cdot \left(n^2 \tau\right)^{\frac{1}{2^k - 1}}\right) = \\
                         &\tOrder\left( \frac{n^{k-1}}{\tau} + n^{3-\frac{1}{2^{k-1}}}+ n^{3-\frac{k-2}{2^k-1}} \cdot \tau^{\frac{1}{2^k-1}}\right)
        \end{split}
    \end{equation*}
    This running time is optimized to $\tOrder\left(n^{3- \frac{1}{2^{k-1}}}\right)$ by setting $\tau = n^{k-4+\frac{1}{2^{k-1}}}$. 
\end{proof}

\begin{proof}[Proof of \cref{thm:attached}]
    Suppose that $G$ is $3$-colored by some coloring $C_G$, as guaranteed by \cref{lem:color-coding}. Compute a coloring $C_H$ satisfying the property in the theorem statement in $\Order(3^k) = \Order(1)$ time. Denote by $k' = |V(H')|$ the size of the remaining subpattern $H'$, so $s = k-k'-1$. 
    The algorithm starts with an embedding phase: it embeds the $s$ vertices that have monochromatic neighborhoods one by one, until no such vertices remain. The embedding is done similarly to that in \cref{sec:H-nice}, so we omit the details. If a triangle is found during the embedding phase, then we are done. Otherwise, the embedding algorithm outputs a list of subgraphs of $G$, such that each subgraph is free from the remaining colored pattern $(H' \cup xyz)$. Moreover, if $G$ contains a triangle, then with probability $1-\Order(1/n)$, one of the subgraphs also contains a triangle. For every subgraph in the list, we invoke the algorithm from \cref{lem:coloredH} with the pattern $H' \cup xyz$, and if a triangle is found, we return it. 

    \paragraph{Correctness}
    Correctness follows from the correctness of the embedding algorithm in \cref{thm:H-nice} and the algorithm from \cref{lem:coloredH}. 
    Both succeed with probability $1-\Order(1/n)$ and hence so does this algorithm.

    \paragraph{Analysis}
    Let us now prove that the running time of the algorithm is $\tOrder(n^{3-\frac{1}{2^{k+s-1}}})$. The running time for the embedding phase depends on the degree thresholds used by the algorithm. Let us denote them by $\Delta_s > \Delta_{s-1} > \cdots > \Delta_1 \geq 1$. Similarly to the analysis in the proof of \cref{thm:H-nice}, the running time we get for the embedding phase is:
    \[
        \tOrder\left(n^2\Delta_s + \frac{n^3 \Delta_{s-1}}{\Delta_{s}} + \frac{n^4 \Delta_{s-2}}{\Delta_s \Delta_{s-1}} + \cdots + \frac{n^{s + 1} \Delta_1}{\prod_{j=2}^{s} \Delta_{j}} \right) = \tOrder\left( \sum_{i=0}^{s-1} \frac{n^{2+i}\Delta_{s-i}}{\prod_{j=0}^{i-1} \Delta_{s-j}} \right).
    \]
    The number of subgraphs produced in this phase is: 
    \[
         \tOrder\left( \frac{n^{s}}{\prod_{j=1}^s \Delta_{j}} \right).
    \]
    Hence, the total running time of the algorithm, including the invocations of \cref{lem:coloredH} for every subgraph, is:
    \[
        \tOrder\left( \frac{n^{s}}{\prod_{j=1}^s \Delta_{j}} \cdot n^{3-\frac{1}{2^{k-1}}} + \sum_{i=0}^{s-1} \frac{n^{2+i}\Delta_{s-i}}{\prod_{j=0}^{i-1} \Delta_{s-j}} \right)= \tOrder\left( \frac{n^{s + 3 - \frac{1}{2^{k-1}}}}{\prod_{j=1}^{s} \Delta_j} + \sum_{i=0}^{s-1} \frac{n^{2+i}\Delta_{s-i}}{\prod_{j=0}^{i-1} \Delta_{s-j}} \right).
    \]
    The running time is optimized to $\tOrder( n^{3-\frac{1}{2^{k+s-1}}} )$ when we set $\Delta_j = n^{1-\frac{1}{2^{k+j-1}}}$ for every $j \in [s]$. Namely, using \cref{eq:exponents}, the left-hand term becomes:
    \[
        \tOrder\left(\frac{n^{s + 3 - \frac{1}{2^{k-1}}}}{\prod_{j=1}^{s} \Delta_j} \right) = \tOrder\left(\frac{n^{s + 3 - \frac{1}{2^{k-1}}}}{n^{s - \sum_{j=1}^s \frac{1}{2^{k+j-1}}}} \right)= \tOrder\left(\frac{n^{s + 3 - \frac{1}{2^{k-1}}}}{n^{s - (\frac{1}{2^{k-1}}-\frac{1}{2^{k+s-1}})}} \right)= \tOrder\left(n^{3-\frac{1}{2^{k+s-1}}} \right), 
    \]
    and every summand on the right-hand term, corresponding to $i \in [0,s-1]$, becomes:
    \[
         \frac{n^{2+i}\Delta_{s-i}}{\prod_{j=0}^{i-1} \Delta_{s-j}} =  \frac{n^{3+i-\frac{1}{2^{k+s-i-1}}} }{n^{i - \sum_{j=0}^{i-1} \frac{1}{2^{k+s-j-1}}}} =  \frac{n^{3+i-\frac{1}{2^{k+s-i-1}}} }{n^{i - (\frac{1}{2^{k+s-i-1}} - \frac{1}{2^{k + s - 1}})}} = n^{3-\frac{1}{2^{k+s-1}}}.
    \]
    Since there are $s = \Order(1)$ summands, we get a total of $\tOrder( n^{3-\frac{1}{2^{k+s-1}}} )$.
\end{proof}

The algorithm is also adapted to list the triangles. 
\begin{proof}[Proof of \cref{thm:listing-attached}]
    Listing during the embedding phase is done the same way as in the algorithm of \cref{thm:listing-nice}. During the sparsification phase, consider an edge $e$ that matches the colors of $xy$, for which a colored copy of $H'$ containing $e$ is sampled. Then, every vertex that forms a triangle with $e$ is contained in that copy, because otherwise, the triangle would form, together with $H'$, a colored copy of $H$ in $G$. Hence, all the triangles containing $e$ can be listed in $\Order(k) = \Order(1)$ time. After the sparsification phase is over, the listing algorithm from \cref{thm:listing-sensitive} is invoked. 
\end{proof}

%% file: sections/odd.tex
In this section, we prove \cref{thm:odd}, showing a Triangle Detection algorithm for $C_{2k+1}$-free graphs. Let us restate the theorem. 

\odd*

\noindent
Our algorithm builds upon an algorithm for a simpler case, where the graph is free from all cycles of lengths $2\ell+1$ for $\ell \in[2,k]$. This algorithm will serve as a basis for our main algorithm for general $C_{2k+1}$-free graphs, which may contain shorter odd cycles.
    The following claim will be crucial in this section. 
    \begin{claim}
        \label{claim:2k}
        For every $\ell \geq 2$, if $G$ is $C_{2\ell+1}$-free, then for every $2\ell$-cycle $C$ in $G$, no edge of $C$ forms a triangle with a vertex outside $C$. 
    \end{claim}
        \begin{proof}
            If a vertex $v \notin V(C)$ is adjacent to the endpoints of an edge $xy \in E(C)$, a $C_{2\ell+1}$ is formed by $v$ and the $(2\ell-1)$-path between $x$ and $y$ in $C$. This leads to a contradiction and completes the proof.
        \end{proof}

\subsection{An algorithm for a simpler case}
    \label{section:girth}
    In this subsection, we prove the following theorem.
    \begin{theorem}
        \label{thm:girth}
        There is a combinatorial Triangle Detection algorithm for graphs that are $C_{2\ell+1}$-free for all $\ell \in[2,k]$, running in time $\Order(m+n^{1+2/k})$.
    \end{theorem}
    \noindent
    The following claim will be useful.
    \begin{claim}
        \label{claim:bfs}
        Suppose that $G$ is $C_{2\ell+1}$-free for all $\ell \in [2,k]$. For every $v \in V(G)$, if there exists an edge $xy \in E(G)$ with $1 \leq \mathrm{dist}(v,x)=\mathrm{dist}(v,y) \leq k$, then there is a triangle $xyz$ for some vertex $z$ such that $\mathrm{dist}(v,z) = \mathrm{dist}(v,x)-1$.
    \end{claim}

    \begin{proof}
        In a BFS tree rooted at $v$, both $x$ and $y$ are in layer $i = \mathrm{dist}(v,x)$. If $x$ and $y$ do not have a common neighbor at layer $i-1$ (which means that $i \geq 2$), then their paths to their lowest common ancestor, together with the edge $xy$, form an odd cycle of length at most $2i+1$, for $2 \leq i \leq k$.  This contradicts the assumption that $G$ is $C_{2\ell+1}$-free for all $\ell \in [2,k]$. Hence, $x$ and $y$ must have a common neighbor $z$ at layer $i-1$, so $xyz$ is a triangle in $G$.
    \end{proof}

    \begin{proof}[Proof of \cref{thm:girth}]
    To find a triangle in $G$, we process its vertices one by one. For each vertex $v \in V(G)$, apply the following procedure:
    \begin{enumerate}[label=(\roman*)]
        \item
            \textbf{Low-degree cleanup.} If $\deg(v) \leq n^{1/k}$, check for triangles containing $v$ by examining every pair of its neighbors. If a triangle is found, report it and halt. Otherwise, delete $v$ from the graph and proceed to the next vertex.
        \item
            \textbf{Constructing an $h$-hop ball.} Construct a BFS tree rooted at $v$ and let $L_i$, for $0 \leq i\leq k $, denote the vertices at distance $i$ from $v$. The BFS is constructed until a \emph{thin} layer $L_h$ is found, which we define as a layer such that $|L_h| < n^{h/k}$. Observe that such a layer must be found for some $h \leq k$, as otherwise the total number of vertices would exceed $n$.
        \item
            \textbf{Searching for a triangle.} If, while constructing layer $L_{i}$ (by scanning the adjacency lists of vertices in $L_{i-1}$), an edge with both endpoints in $L_{i-1}$ is found, then by \cref{claim:bfs}, it must participate in a triangle with some vertex from $L_{i-2}$. Report the triangle and halt.
        \item
            \textbf{Deleting vertices from $G$.} Delete all vertices in $L_{0} \cup \cdots \cup L_{h-2}$.
        \item
            \textbf{Deleting edges between the last two layers.} For every $u \in L_h$ that has two or more neighbors in $L_{h-1}$, delete all edges between $u$ and $L_{h-1}$.
    \end{enumerate}
    If the graph becomes empty at any point, report that $G$ is triangle-free and halt.

    \paragraph{Correctness.}
        For every processed vertex, we delete at least one vertex from the graph, so the algorithm terminates either when a triangle is found or when the graph becomes empty. It is therefore sufficient to show that the algorithm never deletes a vertex or an edge that participates in a triangle. For the low-degree case, it is clear that the deleted vertex is not in a triangle. To see that the deletion in step (iv) is correct, observe that if one of the vertices in $L_{0} \cup \cdots \cup L_{h-2}$ participates in a triangle, then the vertices of that triangle appear in at most two consecutive layers in $L_0 \cup \cdots \cup L_{h-1}$. This implies that one layer contains an edge, which would have been detected in step (iii). Hence, step (iv) is reached only if the deleted vertices do not participate in any triangle. Finally, we show that deleting edges in step (v) is correct. Consider a vertex $u \in L_h$ with two or more neighbors, say $x, y \in L_{h-1}$. The paths in the BFS tree from $x$ and $y$ to their lowest common ancestor, together with the path $x-u-y$, form an even cycle of length $2\ell$ for some $\ell \leq h$. By \cref{claim:2k}, the edges $ux$ and $uy$ do not participate in a triangle with a vertex outside that cycle, because it would form a $2\ell+1$-cycle. Moreover, they do not participate in a triangle with a vertex inside the cycle because the cycle is contained in a BFS, so that could only happen by having an edge $xy$. Such an edge has already been ruled out. 

    \paragraph{Analysis.} We show that the algorithm runs in time $\Order(m + n^{1 + 2/k})$.
        Let $cost_v$ denote the cost of applying the procedure to a vertex $v$. If $v$ is a low-degree vertex, then $cost_v = \Order(n^{2/k})$, so the total cost for low-degree vertices is $\Order(n^{1+2/k})$. Otherwise, $cost_v$ is proportional to the number of edges incident to $L_0 \cup \cdots \cup L_{h_v-1}$, where $h_v \leq k$ is the index of the first thin layer. Let $m_v$ denote the number of edges deleted from $G$ in steps (iv) and (v). So the number of edges incident to $L_0 \cup \cdots \cup L_{h_v-1}$ is bounded by $m_v + |L_{h_v}|$, because the number of deleted edges is $m_v$. The number of undeleted edges is at most $|L_{h_v}|$, as the undeleted edges must be incident to a degree-$1$ vertex in the graph between $L_{h_v}$ and $L_{h_v-1}$. Since $L_{h_v}$ is thin, we get:
        \[
            cost_v = \Order\left( m_v + |L_{h_v}|\right) = \Order\left( m_v + n^{h_v/k} \right).
        \]
        Hence, the total cost for high-degree vertices is bounded by:
        \[
            \Order\left( \sum_{v} m_v + \sum_{v} n^{h_v/k}\right),
        \]
        where the sum is taken only over the high-degree vertices that the algorithm processes. Since each edge is deleted at most once, the first term sums to $\Order(m)$. To bound the second term, observe that the number of vertices deleted by the procedure is at least $|L_{h_v-2}| \geq n^{(h_v-2)/k}$. Therefore, the number of high-degree vertices $v$ for which the procedure stops at layer $h_v=h$ is at most $n / n^{(h-2)/k}$. Hence, we bound the second term by:
        \[
        \sum_{v} n^{h_v/k} = \Order\left(\sum_{h=2}^k n^{h/k} \cdot \frac{n}{n^{(h-2)/k}}\right) = \Order( k n^{1+2/k}) = \Order( n^{1+2/k}).
        \]
        Therefore, the total cost for low-degree and high-degree vertices is
        \[
        \Order\left(n^{1+2/k} + \sum_{v} cost_v\right) = \Order\left( n^{1+2/k} + m + n^{1+2/k} \right) = \Order\left(m + n^{1+2/k}\right).
        \]
    \end{proof}

\subsection{Algorithmic supersaturation}
    A major obstacle in generalizing the above algorithm to the case of $C_{2k+1}$-free graphs is dealing with the edges between the last two layers. In the general case, the only edges that can be deleted between the last two layers are edges that participate in a cycle of length $2k$, as opposed to the previous case where any edge in a cycle of length $2\ell$ for $\ell \leq k$ could be deleted. Finding, or even just proving that there are many edges participating in a $2k$-cycle, is a non-trivial result known as supersaturation. 
    The next tool that we will use in the proof of \cref{thm:odd} is an algorithmic adaptation of a supersaturation result for even-length cycles. 
    \begin{theorem}[Algorithmic supersaturation]
        \label{thm:supersaturation}
        Let $T$ be a tree of height $h \leq k$ with $n$ vertices, and let $L$ be its set of leaves, all at depth $h$. Consider a bipartite graph $F$ with parts $L$ and a set $W$ of at most $n$ vertices. Let $G$ denote the graph $G = T \cup F$. Then there is an algorithm that computes a set of pairs $X$, containing at least $e(F)/k - \Order(k^2 (|L| + |W|))$ pairs of the form $(uv,C)$, where $uv \in E(F)$, and $C \subseteq G$ is a $2k$-cycle containing $uv$. The running time of the algorithm is 
        $\Order( e(F) \log n + n )$.
    \end{theorem}
    As mentioned earlier, this proof is an algorithmic adaptation of \cite[Section 4]{jiang2020supersaturation}. For a self-contained presentation, we include all the details here with slight variations from the original proof. The first lemma is the algorithmic analog of \cite[Lemma 4.4]{jiang2020supersaturation}.
    \begin{lemma}
        \label{lem:4.4}
        Let $T,F,W,L$ be as in the statement of \cref{thm:supersaturation}. Additionally, suppose that $\deg(w) \geq 2k^2$ for every $w \in W$. There is an algorithm that computes a function $f \colon W \to V(T)$, such that for every $w \in W$, if $f(w)$ is at depth $i$, then:
        \begin{itemize}
            \item $|N(w) \cap V(T_{f(w)})| \geq |N(w)| - i k$
            \item For every child $z$ of $f(w)$, $|N(w) \cap V(T_z)| < |N(w)| - (i+1)k$.
        \end{itemize}
        Moreover, the running time of the algorithm is $\Order( e(F) \log n )$.
    \end{lemma}
    \begin{proof}
        As a preprocessing stage, for every $\ell \in L$, store the path $P_\ell$ from the root of $T$ to $\ell$, indexed so that $P_\ell[0]$ is the root and $P_\ell[h] = \ell$. This can be computed in $\Order(kn)$ time in a top-down fashion. The algorithm then proceeds in $h+1$ phases, from $i=0$ to $h$. In phase $i$, it performs the following steps:
        \begin{enumerate}[label=(\roman*)]
            \item For every $\ell \in L$, set $P_{\ell}[i]$ as the key of $\ell$.
            \item For every $w \in W$, let $c(w)$ be the most common key in $N(w)$. This is computed by sorting the neighbors of $w$ according to their keys.
            \item For every $w \in W$, if the frequency of $c(w)$ in the sorted list is at least $|N(w)| - i k$, then set $f(w) \leftarrow c(w)$.
        \end{enumerate}

        \paragraph{Correctness.} Let us prove that at the end of the last phase, for every $w \in W$, a vertex $f(w)$ is computed that satisfies the properties in the statement.
        First, observe that $f(w)$ is defined for every vertex $w \in W$, because in the first phase ($i=0$), all vertices in $L$ have the same key, the root of $T$, say $x$. Thus, for every $w$, $c(w)=x$, and in particular, the condition $|N(w) \cap V(T_x)| = |N(w)| \geq |N(w)| - 0 \cdot k$ is met, so $f(w)$ is set to $x$.
        Now, suppose that at the end of the algorithm, $f(w) = y$ for some $y \in V(T)$. Suppose $y$ has depth $i$. We first show that $i<h$. By the definition of the procedure, $|N(w) \cap V(T_y)| \geq |N(w)| - i k \geq 2k^2-ik$. If $y$ were at depth $h$, then $y \in L$, so $|V(T_y)|=1$. This contradicts $|N(w) \cap V(T_y)| \geq 2k^2-k^2=k^2 > 1$ (since $k \geq 2$). Hence, $y$ has children in the tree. For every child $z$ of $y$, it must be that $|N(w) \cap V(T_z)| < |N(w)| - (i+1)k$, because otherwise, $f(w)$ would have been updated to $z$ in phase $i+1$. Hence, $y$ satisfies both properties.

        \paragraph{Analysis.} We now analyze the running time.
        The preprocessing stage takes $\Order(k n )$ time. There are at most $k+1$ phases. In each phase, the adjacency lists of all $w \in W$ are sorted (e.g., using Merge Sort). The time for each phase is therefore $\Order(\sum_{w \in W} \deg(w) \log (\deg(w))) = \Order( e(F) \log n)$. Since $k$ is a constant, the total time is $\Order( e(F) \log n  + n)$.
    \end{proof}
    
    Now we are ready to prove the algorithmic supersaturation theorem. 
    \begin{proof}[Proof of \cref{thm:supersaturation}]
        Throughout the proof, we use $L_i$ to denote the set of vertices at distance $i$ from the root of $T$. In particular, $L = L_h$.  
        The algorithm is as follows: 
        \begin{enumerate}[label=(\roman*)]
            \item
                \textbf{Cleanup.} Delete vertices in $W$ of degree less than $2k^2$, until no such vertices remain. Denote the remaining vertices by $W'$. Next, apply the algorithm from \cref{lem:4.4} to compute $f \colon W' \to V(T)$. For each $w \in W'$, delete all edges from $w$ to $N(w) \setminus T_{f(w)}$. Denote the resulting subgraph of $F$ by $F'$. For every $i \in [0,h-1]$, let $W_i = \set{w \in W' \mid f(w) \in L_i }$ and let $F_i \subseteq F'$ be the subgraph induced by $W_i \cup L$. Compute $j \in [0,h-1]$ to be such that $e(F_j)$ is maximized. 
            \item
                \textbf{Decomposing $\mathbf{F_j}$.} For every $y \in L_j$, let $Z_y = \{ w \in W_j \mid f(w) = y \}$ and let $S_y = N_{F_j}(Z_y)$. Let $G_y$ be the induced subgraph $F_j[Z_y \cup S_y]$. Note that for $y \neq y'$, $G_y$ and $G_{y'}$ are vertex-disjoint subgraphs of $F_j$.\footnote{$Z_y$ and $Z_{y'}$ are disjoint by definition. $S_y$ and $S_{y'}$ are disjoint because they belong to subtrees rooted at different vertices on the same level.} 
            \item 
                \textbf{Data structure for non-membership listing.} For every graph $G_y$, we construct a data structure to answer queries of the following form: Given $w \in Z_y$ and a child $z \in V(T_y)$ of $y$, list the neighbors of $w$ in $S_y$ that are \emph{not} in the subtree $T_z$. This is implemented as follows for every $G_y$: (1) Sort the adjacency list of every $w \in Z_y$, using as keys the children of $y$ that are ancestors of the neighbors of $w$. (2) Given a query, output the neighbors of $w$ as they appear in the sorted list, skipping the block of neighbors whose key is $z$.
            \item
                \textbf{Computing $O(k)$-cores.} For every $y \in L_j$, compute a $2k$-core of $G_y$ by invoking the algorithm from \cref{claim:core}. Denote the core by $G_y'$, and let $Z_y' \subseteq Z_y$ and $S_y' \subseteq S_y$ be the vertex parts of $G_y'$. Keep a copy of $G_y$ for later use.
            \item
                \textbf{Finding many $2k$-cycles.} For every non-empty core $G_y'$, and for every edge $\ell w \in E(G_y')$ where $\ell \in S_y'$ and $w \in Z_y'$, we compute a $2k$-cycle containing $\ell w$ as follows. Compute a path $Q$ of length $2k-2(h-j)-2$ in $G_y'$, starting from $w$ and avoiding $\ell$. This can be done in a $2k$-core by repeatedly picking an unvisited neighbor. Let $q \in Z_y'$ be the final vertex in that path and let $z$ be the child of $y$ that is the ancestor of $\ell$. We pick a neighbor $\ell' \in N_{G_y}(q)$ satisfying $\ell' \notin \{\ell\} \cup V(Q) \cup V(T_z)$. To see that such a vertex exists, recall that since $f(q) = y$, then: 
                \begin{itemize}
                    \item $q$ has at least $|N_G(q)| - jk$ neighbors in $V(T_y)$. Hence, $|N_{G_y}(q)| \geq |N_G(q)| - jk$. 
                    \item $q$ has fewer than $|N_G(q)| - (j+1)k$ neighbors in $V(T_z)$.
                \end{itemize}
                Let us count how many more forbidden neighbors $q$ has in $G_y$, beyond $V(T_z)$. 
                Note that the number of vertices of $Q$ in $S_y$ is exactly $k-(h-j)-1$. Thus, the number of forbidden neighbors $q$ has in $S_y$ is at most:
                \begin{equation*}
                    \begin{split}
                        |\{\ell\} \cup (V(Q)\cap S_y) \cup (N_{G_y}(q) \cap V(T_z))| &< 1 + (k-(h-j)-1) + (|N_G(q)| - (j+1)k) \\
                        &= |N_G(q)| - jk - k + k - (h-j) \\
                        &\leq |N_G(q)| - jk \\
                        &\leq |N_{G_y}(q)|.
                    \end{split}
                \end{equation*}
                Since the number of forbidden neighbors is strictly less than the number of available neighbors $q$ has in $G_y$, such an $\ell'$ exists. We find it by querying the data structure from (iii) for neighbors of $q$ not in $T_z$, until one not in $\{\ell\} \cup V(Q)$ is found.
                Let $C$ be the cycle formed by the path $\ell w \circ Q \circ q\ell'$ and the paths in $T$ from $\ell$ and $\ell'$ to their lowest common ancestor, which is $y$. The length of this cycle is $(1 + (2k-2(h-j)-2) + 1) + (h-j) + (h-j) = 2k$. Add the pair $(\ell w,C)$ to the output set $X$.
        \end{enumerate}
        It is clear that the cycle $C$ is a $2k$-cycle containing $\ell w$. It remains to prove the bound on $|X|$ and analyze the running time.
        \paragraph{Correctness.} We prove that $|X| \geq e(F)/k - \Order(k^2 (|L|+|W|))$.
            In the cleanup stage, we delete at most $\Order(|W|k^2)$ edges incident to vertices in $W$ with fewer than $2k^2$ neighbors in $L$. Then, by the definition of $f(w)$, the number of edges removed from each $w$ that are not incident to $T_{f(w)}$ is at most $ik \leq k^2$. In total, at most $\Order(|W| k^2 )$ edges are removed so far, hence $e(F') \geq e(F) - \Order(|W| k^2)$. Observe that $e(F') = \sum_{i = 0}^{h-1} e(F_i)$ because the $F_i$'s are edge disjoint. Hence, $e(F_j) \geq e(F)/k - \Order(|W| k^2)$. In the final stage, we find a cycle for every edge in the $2k$-cores. The number of edges not in the $2k$-cores is at most $\sum_y (2k-1)|V(G_y)| = \sum_y (2k-1)(|S_y| + |Z_y|) =  \Order(k (|L| + |W|)) = \Order(k^2(|L| + |W|))$. Thus, $|X| \geq e(F)/k -  \Order(k^2(|L| + |W|))$. 

        \begin{claim}[Analysis]
            \label{claim:analysis_supersaturation}
            The running time of the algorithm is $\Order(n + e(F) \log n)$.
        \end{claim}
        \begin{proof}
            Step (i) takes $\Order(n+ e(F))$ plus $\Order(e(F) \log n + n )$ for the call to the algorithm of \cref{lem:4.4}. Steps (ii), (iii), and (iv) take $\Order(n+e(F)\log n)$ time, dominated by sorting. Finally, for every edge in a core, we find a cycle. This involves building a path of length $\Order(k)$, which takes $\Order(k^2)$ time, and querying the data structure $\Order(k)$ times. Since $k$ is a constant, the cost per edge is constant. The total time is dominated by the sorting steps, resulting in $\Order(n + e(F)\log n )$.
        \end{proof}
    \end{proof}

\subsection{Proof of \cref{thm:odd}}
We now adapt the algorithm from \cref{section:girth} to the setting of $C_{2k+1}$-free graphs. This case is more complicated because there can be edges within BFS layers, and edges between the final two layers cannot be handled in the same way as before. We will need the following generalization of the Chiba-Nishizeki algorithm.
        \begin{lemma}[Generalized Chiba-Nishizeki]
            \label{lem:generalized-chiba}
            Let $G$ be a graph and let $S \subset V(G)$ be a given set such that $G[S]$ is $d$-degenerate. There is an algorithm that detects whether there is a triangle in $G$ with at least two vertices in $S$ in time $\Order( d \cdot e(S,V\setminus S))$.
        \end{lemma}
        \begin{proof}
            Process the vertices in $S$ by repeatedly picking the vertex of minimum degree in $G[S]$. For each such vertex $v \in S$, check if it participates in a triangle by examining every combination of a neighbor in $S$ and a neighbor in $V(G)$. If a triangle is not found, delete $v$ from the graph (and from $S$). Since $G[S]$ is $d$-degenerate, there is always a vertex with at most $d$ neighbors in $S$. Therefore, the total running time is bounded by 
            \[
                \Order\left(\sum_{v \in S} \deg_S(v) \deg_{V(G) \setminus S}(v)\right) = \Order\left(d \cdot  e(S,V(G) \setminus S)\right),
            \]
            where $\deg_S(v)$ and $\deg_{V(G) \setminus S}(v)$ are the number of neighbors $v$ has in $S$ and $V(G) \setminus S$, respectively. 
        \end{proof}
        
        We will also use a lemma by Keevash et al., rephrased for our purposes.
        \begin{lemma}[Lemma 3.8 in ~\cite{keevash2013conjecture}, rephrased]
            \label{lem:keevash}
            Let $L_i$ be the set of vertices at distance $i$ from some vertex $v$ in a $C_{2k+1}$-free graph $G$, where $k \geq 2$. If $i \leq k$, then the degeneracy of $G[L_i]$ is $\Order(k)$.
        \end{lemma}
        We remark that the original statement only contains an upper bound of $\Order(k)$ on the average degree of $G[L_i]$, and not its degeneracy. However, their proof actually shows the stronger statement about the degeneracy. Let us now explain how to adapt the algorithm from \cref{thm:girth}.

    \begin{proof}[Proof of \cref{thm:odd}]
        The algorithm is as follows.
        \begin{enumerate}[label=(\roman*)]
            \item
                \textbf{Low-degree cleanup.} This stage remains the same as in the proof of \cref{thm:girth}.
            \item
                \textbf{Constructing an $h$-hop ball.} This stage remains the same as in the proof of \cref{thm:girth}.
            \item
                \textbf{Searching for a triangle.} For every $i \leq h-2$, check if $G[L_i \cup L_{i+1}]$ contains a triangle. By \cref{lem:keevash}, $G[L_i]$ and $G[L_{i+1}]$ have degeneracy $\Order(k)$. Apply the generalized Chiba-Nishizeki algorithm (\cref{lem:generalized-chiba}) on $G[L_i \cup L_{i+1}]$ twice: once with $S = L_i$ and once with $S = L_{i+1}$. If a triangle is found, report it and halt.
            \item
                \textbf{Deleting vertices from $G$.} This stage remains the same as in the proof of \cref{thm:girth}.
            \item
                \textbf{Deleting edges between the last two layers.} Apply the algorithm from \cref{thm:supersaturation} to the graph on $L_{h-1} \cup L_{h}$ (with the BFS tree providing the tree structure) to get a set of edge-cycle pairs $X$. For every $(uv,C) \in X$, check if $uv$ participates in a triangle with another vertex in $C$. If so, report it and halt. Otherwise, delete $uv$ from $G$.
        \end{enumerate}
        \paragraph{Correctness.} We prove that if $G$ contains a triangle, the algorithm finds it. It suffices to show that stage (iii) finds any triangle incident to $L_0 \cup \cdots \cup L_{h-2}$, and that in stage (v), every edge that gets deleted does not participate in a triangle.

            For stage (iii), recall that a triangle can only be incident to at most two consecutive layers. Hence, if a triangle is incident to $L_0 \cup \cdots \cup L_{h-2}$, it must lie fully within $G[L_i \cup L_{i+1}]$ for some $i \leq h-2$. Such a triangle must have at least two vertices in $L_i$ or at least two in $L_{i+1}$, so the generalized Chiba-Nishizeki algorithm is guaranteed to find it.

            For stage (v), consider a pair $(uv,C) \in X$, where $C$ is a $2k$-cycle containing $uv$. By \cref{claim:2k}, $uv$ may only participate in a triangle with another vertex in $C$, and such triangles are checked before deleting $uv$. 
            \paragraph{Analysis} We now prove that the running time of the algorithm is $\Order((m + n^{1+2/k})\log n)$.
            The analysis is similar to that in the proof of \cref{thm:girth}, with the difference being stages (iii) and (v). Suppose that $v$ is a vertex of degree greater than $n^{1/k}$ and let $h_v$ be the height of the BFS rooted at $v$. We see that the cost of $v$ is proportional to the number of edges incident to $L_0 \cup \cdots \cup L_{h_v-1}$, as the computation time for stages (iii) and (v) is near-linear in the number of those edges. To bound the total cost of all the vertices, we again split the cost into two:
            \begin{itemize}
                \item \textbf{Deleted edges:} every edge that is incident to $L_0 \cup \cdots \cup L_{h_v-2}$, or appears in $E(L_{h_v-2},L_{h_v-1})$, gets deleted in step (iv). Moreover, every edge that appears in $X$ in stage (v) gets deleted. 
                \item \textbf{Undeleted edges:} edges in $E(L_{h_v-1},L_{h_v})$ that did not appear in $X$. 
            \end{itemize}
            By \cref{thm:supersaturation}, the number of edges that appear in $X$ is at least $e(L_{h_v},L_{h_v-1})/k - \Order(k^2(|L_{h_v-1}| + |L_{h_v}|))$. Hence, denoting by $m_v$ the number of deleted edges, we have:
            \[
                e(L_{h_v},L_{h_v-1}) \leq  k m_v + \Order(k^3(|L_{h_v-1}| + |L_{h_v}|)) = \Order(m_v + |L_{h_v-1}| + |L_{h_v}|). 
            \]
            The number of undeleted edges is at most $e(L_{h_v},L_{h_v-1})$. We can therefore write:
            \[
                cost_v = \tOrder(m_v + e(L_{h_v},L_{h_v-1})) = \tOrder(m_v  + |L_{h_v}| + |L_{h_v-1}|).
            \]
            Note that the term $|L_{h_v-1}|$ is accounted for by $m_v$ because $m_v$ contains at least that many edges in $E(L_{h_v-2},L_{h_v-1})$. For the $|L_{h_v}|$ term, it is at most $n^{h_v/k}$ because the last layer is thin. Hence, we get: 
            \[
                cost_v = \tOrder\left( m_v  + n^{h_v/k} \right).
            \]
            By a similar calculation as in the proof of \cref{thm:girth}, we get the bound:
            \[
                \tOrder\left(m + n^{1+2/k}\right).
            \]
    This concludes the proof of \cref{thm:odd}.
\end{proof}

Let us now finish this subsection with a simple generalization to listing triangles.
\listingodd*
\begin{proof}[Proof of \cref{thm:listing-odd}]
    To list the triangles instead of reporting just one, we observe that the generalized Chiba-Nishizeki algorithm from \cref{lem:generalized-chiba} can be easily modified to list all the triangles with an edge in $S$ it encounters, in the same running time. So during a phase of the algorithm, all triangles in the $h$-hop ball incident to layers $L_0,\ldots,L_{h-2}$ can be listed. The edges that get deleted in the final step of the algorithm do not participate in a triangle, and hence the algorithm correctly lists all triangles in the graph. The analysis remains the same. 
\end{proof}

%% file: sections/propositions.tex
In this section, we prove \cref{prop:smallpatterns} and \cref{prop:3k}. 

\subsection{Proof of \cref{prop:3k}}
    Let us recall the proposition statement. 
    \divisible*
    \begin{proof}
        Consider any proper $3$-coloring of $H$ with colors from $\set{0,1,2}$. Suppose, for contradiction, that there is no vertex with a monochromatic neighborhood. Construct a maximal path $v_0, v_1 ,\ldots$ such that $v_0$ is some arbitrary vertex of color $0$ (which must exist by the hypothesis), and for $i > 0$, the color of $v_i$ is $i \mod 3$. Denote by $v_\ell$ the last vertex in the path. By the hypothesis, it has a neighbor $u$ of color $(\ell+1)\mod 3$. However, by maximality of the path, $u$ must be some vertex $v_j$ for $j < \ell$. Since the color of $u$ is $(\ell + 1) \mod 3$ and the color of $v_j$ is $j \mod 3$, we must have the congruence $j \equiv (\ell+1) \mod 3$. The length of the cycle consisting of the $v_j$-to-$v_\ell$ path and the edge $v_\ell v_j$ is $\ell-j+1$. Given the congruence, we have $\ell-j+1 \equiv 0 \mod 3$. 
    \end{proof}

\subsection{Proof of \cref{prop:smallpatterns}}
    Let us recall the proposition statement.
    \smallpatterns*
    \begin{proof}
        It suffices to prove the claim for patterns of size $8$ because any smaller pattern that does not fall into one of the categories can be extended to $8$ vertices by attaching some vertices. 
        It turns out that all patterns of size $8$, except for eight patterns, fall into the first three categories. The proof of this claim is computer-aided. Our computer program \cite{code} enumerates all patterns of size $8$ that do not fall into the lower bound category (\cref{thm:lowerbounds}), i.e., patterns that are $3$-colorable and contain at most one triangle. 
        For each of these patterns, the program checks whether it admits a nice coloring. The final output contains those eight patterns that are not in the lower bound category and do not admit a nice coloring. 

        Let us show that the remaining eight patterns fall into the fourth category. Among them, four are obtained by attaching a vertex, or adding an isolated vertex, to the $7$-vertex pattern in \cref{fig:7vertices}. Hence, it suffices to show that this pattern falls into the fourth category. The illustrated coloring shows this. 
        \begin{figure}[h]
            \centering
            \includegraphics[scale=0.8]{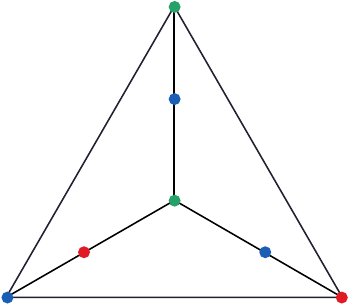}
            \caption{A $7$-vertex pattern which does not fall into the first three categories. The illustrated coloring admits a vertex with a monochromatic neighborhood. After removing it, the remaining pattern is a triangle attached to a nicely colored pattern. This proves that the pattern falls into the fourth category.}
            \label{fig:7vertices}
        \end{figure}
        \noindent
        The four remaining patterns are illustrated in \cref{fig:combinedfigures}, together with colorings demonstrating that they fall into the fourth category. 
        \begin{figure}[h!]
            \centering
            \begin{subfigure}[h]{0.48\textwidth}
                \centering
                \includegraphics[scale=0.8]{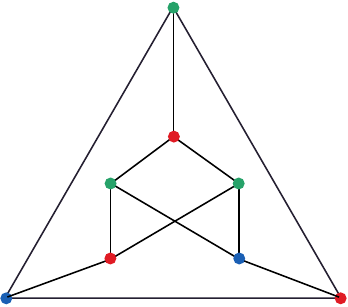}
                \caption{}
                \label{fig:8vertices1}
            \end{subfigure}
            \hfill 
            \begin{subfigure}[h]{0.48\textwidth}
                \centering
                \includegraphics[scale=0.8]{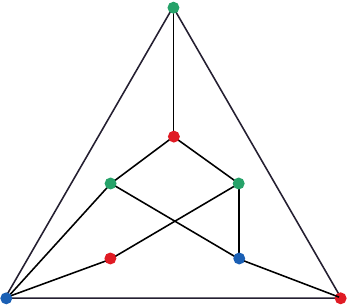}
                \caption{}
                \label{fig:8vertices2}
            \end{subfigure}
            \hfill 
            \begin{subfigure}[h]{0.48\textwidth}
                \centering
                \includegraphics[scale=0.8]{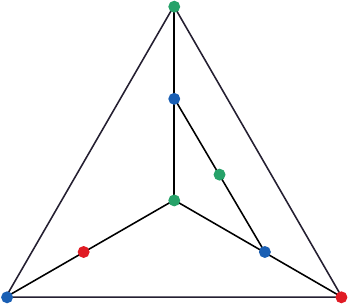}
                \caption{}
                \label{fig:8vertices3}
            \end{subfigure}
            \begin{subfigure}[h]{0.48\textwidth}
                \includegraphics[scale=0.8]{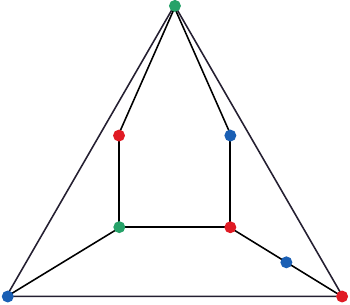}
                \caption{}
                \label{fig:8vertices4}
            \end{subfigure}
            \caption{The four patterns on $8$ vertices which are not extensions of the pattern from \cref{fig:7vertices}. In each of these colored patterns, after removing one or two vertices with a monochromatic neighborhood, the remaining pattern is a triangle attached to a nicely colored pattern. This shows that the patterns fall into the fourth category.}
            \label{fig:combinedfigures}
        \end{figure}
    \end{proof}

%% file: sections/conclusions.tex
Our work is a first step in the systematic study of structured instances of combinatorial Triangle Detection. It makes progress towards a dichotomy theorem for $H$-free graphs, with all patterns of size at most $8$ classified. Several open questions arise. The most pressing one is to complete the dichotomy for all patterns. We put forth the following hypothesis:

\begin{hypothesis*}[Dichotomy for Triangle Detection in $H$-Free Graphs]
The Triangle Detection problem in $H$-free graphs is:
\begin{itemize}
\item in $\Order(n^{3-\eps})$ time via a combinatorial algorithm for some $\eps>0$, if $H$ is $3$-colorable and contains at most one triangle; and
\item as hard as Triangle Detection in general graphs, otherwise.
\end{itemize}
\end{hypothesis*}

Note that this hypothesis is making two non-obvious claims: First, that there is a dichotomy by which each pattern $H$ can be classified as either ``easy'' or \emph{triangle-hard}; i.e., that there is no $H$ for which Triangle Detection in $H$-free graphs requires cubic time but there is no reduction from Triangle Detection in general graphs to Triangle Detection in $H$-free graphs. Second, it asserts that the only two conditions that can make a pattern triangle-hard are having a chromatic number greater than $3$ or having more than one triangle.

To prove the hypothesis, one has to extend our algorithms to more patterns, and to disprove it, one has to extend our lower bounds. We believe that the former is more likely. In fact, our current algorithmic techniques of ``embedding'' subpatterns by sampling vertices may be sufficient if new, more intricate analysis tools are developed.

\paragraph{Other open questions.} From the lower bound perspective, our current understanding is that some patterns admit an $n^{3-o(1)}$ lower bound, while any other pattern only admits the trivial $\Omega(n^2)$ lower bound. It would be interesting to prove a super-quadratic lower bound (i.e., $\Omega(n^{2+\eps})$ for some $\eps >0$) for Triangle Detection in $H$-free graphs for some pattern $H$ that admits a subcubic algorithm. This also addresses the goal of optimizing the running times, which was not our main focus in the paper. 

Another open question is about the necessity of randomness in our algorithms. In the $H$-free setting, it seems possible to avoid randomness by using deletions, as is done in our specialized algorithm for $C_{2k+1}$-free graphs. However, in the $H$-sensitive setting, our algorithm crucially depends on randomness as it needs to quickly find a neighbor of the triangle that participates in a small number of copies of $H$. Derandomizing this component seems more difficult. 

Finally, this work may be further developed, for example, by considering $k$-Clique Detection in $H$-free graphs or directed versions of the problem. Another equally important question is about Triangle Detection in induced $H$-free graphs, which we did not address in this work. 